\documentclass[submission,copyright,creativecommons]{eptcs}
\providecommand{\event}{GandALF 2018} % Name of the event you are submitting to
\usepackage{underscore}           % Only needed if you use pdflatex.

\usepackage{verbatim}

\usepackage{amsmath}
\usepackage{amssymb}
\usepackage{amsthm}
\usepackage{stmaryrd} % semantic brackets
\usepackage{wasysym}
\usepackage{mathtools} % := symbol
\usepackage{bm}

\usepackage[noend]{algpseudocode}
\usepackage{algorithm}

\usepackage{multirow}

\usepackage{cite}
\usepackage{hyperref}

\usepackage{tikz}
\usetikzlibrary{arrows,automata,positioning}
\usetikzlibrary{decorations.pathmorphing}
\usepgflibrary{shapes.geometric}

\usepackage{petrigames}
\usepackage{subfigure}

\newtheorem{theorem}{Theorem}
\newtheorem{lemma}{Lemma}
\theoremstyle{remark}
\newtheorem{example}{Example}

\newcommand{\smallparagraph}[1]{\medskip\noindent\textbf{#1}\;}

\newcommand{\set}[1]{{\{#1\}}}
\newcommand{\tuple}[1]{{\langle#1\rangle}}
\newcommand{\overlineindex}[1]{\overline{#1\mkern-3mu}\mkern 3mu}
\newcommand{\ldot}{\mathpunct{.}}

\newcommand{\card}[1]{|#1|}
\newcommand{\dom}{\mathrm{dom}}

\newcommand{\bigO}{\mathcal{O}}

\newcommand{\bool}{\mathbb{B}}

\newcommand{\limplies}{\rightarrow}

\newcommand{\quant}{\mathop{Q}}
\newcommand{\dep}{\mathit{dep}}
\newcommand{\var}{\mathit{var}}
\newcommand{\lit}{\mathit{lit}}

\newcommand{\assignment}{\alpha}
\newcommand{\Assignment}{\mathcal{A}}
\newcommand{\subassign}{\sqsubseteq}
\newcommand{\assunion}{\sqcup}
\newcommand{\free}{\mathit{free}}
\newcommand{\subf}{\mathit{sf}}
\newcommand{\dsubf}{\mathit{dsf}}
\newcommand{\boolf}{\mathcal{B}}

\newcommand{\sftype}{\mathit{type}}
\renewcommand{\models}{\vDash}
\newcommand{\nmodels}{\nvDash}
\newcommand{\compat}{\prec}

\newcommand{\enc}{\mathit{enc}}
\newcommand{\encout}{\mathit{out}}

\newcommand{\sat}{\textsc{sat}}
\newcommand{\SAT}{\mathrm{SAT}}
\newcommand{\UNSAT}{\mathrm{UNSAT}}

\newcommand{\res}{\mathit{result}}

\newcommand{\quabs}{\textsc{QuAbS}}

\newcommand{\cqesto}{\textsc{cQESTO}}
\newcommand{\caqe}{\textsc{CAQE}}
\newcommand{\ghostq}{\textsc{GhostQ}}
\newcommand{\qfun}{\textsc{QFUN}}
\newcommand{\qute}{\textsc{Qute}}

\newcommand{\bosy}{\text{BoSy}}

\newcommand{\absproof}{\mathcal{P}}
\newcommand{\invabs}{\mathcal{C}}

\newcommand{\absqbf}{\textsc{abstraction-qbf}}

\title{Solving QBF by Abstraction\thanks{Supported by the German Research Foundation (DFG) Grant Petri Games (No.\ 392735815) and by the European Research Council (ERC) Grant OSARES (No.\ 683300).}}

\author{%
Jesko Hecking-Harbusch\institute{Reactive Systems Group\\Saarland University}%
\and%
Leander Tentrup\institute{Reactive Systems Group\\Saarland University}%
}

\def\titlerunning{Solving QBF by Abstraction}
\def\authorrunning{J.~Hecking-Harbusch and L.~Tentrup}

\begin{document}

\maketitle

\begin{abstract}
Many verification and synthesis approaches rely on solving techniques for quantified Boolean formulas~(QBF).
Consequently, solution witnesses, in the form of Boolean functions, become more and more important as they represent implementations or counterexamples.
We present a recursive counterexample guided abstraction and refinement algorithm~(CEGAR) for solving and certifying QBFs that exploits structural reasoning on the formula level.
The algorithm decomposes the given QBF into one propositional formula for every block of quantifiers that abstracts from assignments of variables not bound by this quantifier block.
Further, we show how to derive an efficient certification extraction method on top of the algorithm.
We report on experimental evaluation of this algorithm in the solver $\quabs$ (Quantified Abstraction Solver)  which won the most recent QBF competition (QBFEVAL'18).
Further, we show the effectiveness of the certification approach using synthesis benchmarks and a case study for synthesizing winning strategies in Petri Games.

\end{abstract}

%===============================================================================
\section{Introduction} \label{sec:introduction}
%===============================================================================

Synthesis is the task to produce correct-by-design implementations from formal specifications.
This allows the developer to focus on what to achieve in form of the specification instead of focussing on how to implement requirements.
The synthesis task is usually formulated as a two-player game between the \emph{system} player whose objective is to satisfy the specification and the \emph{environment} player who tries to falsify the specification.
There are many variants of such games in literature, suitable for different types of systems, such as synchronous, asynchronous, and distributed ones, and for different kinds of objectives, such as safety objectives, $\omega$-regular winning conditions, and beyond.
Determining the winner of a synthesis game, which is equivalent to the answer whether the underlying specification is \emph{realizable}, gives us the knowledge of whether or not an implementation exists that satisfies the specification.
In the best case, we can directly construct an implementation from a winning strategy of the synthesis game.

In this paper, we consider the satisfiability problem of quantified Boolean formulas~(QBF), which can be formulated as a game between the \emph{existential} and \emph{universal} player, controlling the existential and universal quantifiers, respectively.
QBF has been used to encode the realizability problem for many of the specifications and games mentioned before, such as symbolically represented safety games~\cite{conf/vmcai/BloemKS14}, the LTL realizability problem~\cite{conf/tacas/FaymonvilleFRT17}, distributed and fault-tolerant synthesis~\cite{conf/tacas/FinkbeinerT14,journals/corr/FinkbeinerT15}, and asynchronous systems using Petri games~\cite{conf/birthday/Finkbeiner15}.
As a side-effect of those encodings, a certification of the QBF solving result in many cases directly corresponds to winning strategies and implementations.
\emph{QBF certification} is the task to extract \emph{Skolem} functions for the existential quantifiers of true QBFs and \emph{Herbrand} functions for the universal quantifiers of false QBFs.

Despite its benefits, QBF certification is a weak spot of current solving algorithms.
There are a number of works in the literature~\cite{journals/fmsd/BalabanovJ12,conf/sat/NiemetzPLSB12,conf/fmcad/HeuleSB14,conf/fmcad/RabeT15} for certifying QBFs given in conjunctive normal form (CNF), but in practice it involves performance penalties due to non-applicable solving optimizations and limited preprocessing\footnote{In QBFEVAL'16, the best certifying QBF solver has solved less than half of the number of instances solved by the best non-certifying solver~\cite{conf/sat/Pulina16}. In the following two iterations of QBFEVAL up to this paper, certification has not been evaluated.}.
We believe that certification must be treated as a first-class citizen of QBF solving and show that it is possible to have competitive performance and solution extraction at the same time.
A crucial approach to this goal is that we consider formulas in negation normal form~(NNF) instead of CNF.
Using (the less restrictive) NNF, the QBF solving problem becomes dual with respect to negation, removing the inherent imbalance of CNF solving algorithms~\cite{conf/epia/JanotaM17}.

We present a counterexample guided abstraction and refinement (CEGAR) algorithm for solving quantified Boolean formulas that exploits the propositional formula's structure.
The algorithm decomposes the given QBF into one propositional formula for every maximal block of consecutive quantifiers of the same type.
We call this formula an abstraction, because it \emph{abstracts} from assignments of variables that are not bound by this block of quantifiers of the same type.
We use special \emph{interface variables} to communicate \emph{assumptions} (outer-to-inner quantifier) and \emph{learned information} (inner-to-outer quantifier) during solving.
Further, we use a SAT solver as an oracle to generate new abstraction entries and to provide us with witnesses for unsatisfiable queries.
Given a QBF, the algorithm proceeds by generating a candidate solution using a SAT solver and the abstraction.
Then, this candidate is verified (or refuted) recursively and, depending on the result, the abstraction is refined.
We introduce a new element to QBF refinement algorithms by maintaining, for every quantifier block, a \emph{dual abstraction} and make twofold use of it: It provides a method for optimizing abstraction entries and it is used to translate counterexamples from one quantifier block to another.

To sum up, this paper makes the following contributions:
\begin{itemize}
  \item We provide a counterexample guided abstraction and refinement (CEGAR) algorithm for solving QBFs in negation normal form.
  \item We describe an efficient certification approach and evaluate it on synthesis benchmarks where implementations can be obtained from certificates.
  \item As a case study, we show how to make use of the certification feature to build strategies representing implementations and counterexamples for Petri games.
\end{itemize}

%===============================================================================
\section{Quantified Boolean Formulas} \label{sec:qbf}
%===============================================================================

A quantified Boolean formula~(QBF) is a propositional formula over a finite set of variables $\mathcal{X}$ extended with quantification.
The syntax is given by the grammar
\begin{equation*}
  \varphi \coloneqq x \mid \neg \varphi \mid \varphi \lor \varphi \mid \varphi \land \varphi \mid \exists x \ldot \varphi \mid \forall x \ldot \varphi \enspace,
\end{equation*}
where $x \in \mathcal{X}$. 
For readability, we lift the quantification over variables to the quantification over sets of variables and denote a maximal consecutive block of quantifiers of the same type $\forall x_1\ldot \forall x_2\ldot \cdots\forall x_n\ldot\varphi$ by $\forall X\ldot\varphi$ and $\exists x_1\ldot\exists x_2\ldot\cdots\exists x_n\ldot\varphi$ by $\exists X\ldot\varphi$, accordingly, where $X=\set{x_1,\dots,x_n}$. 

Given a subset of variables $X \subseteq \mathcal{X}$, an \emph{assignment} of $X$ is a function $\assignment : X \rightarrow \bool$ that maps each variable $x \in X$ to either true ($1$) or false ($0$).
When the domain of $\alpha$ is not clear from context, we write $\alpha_X$.
A partial assignment $\beta : X \rightarrow \bool \cup \set{\bot}$ may additionally set variables $x \in X$ to an undefined value~$\bot$.
We say that $\beta$ is \emph{compatible} with $\alpha$, written $\beta \subassign \alpha$, if they have the same domains ($\dom(\alpha) = \dom(\beta)$) and $\alpha(x) = \beta(x)$ for all $x \in \dom(\alpha)$ where $\beta(x) \neq \bot$.
For two assignments $\alpha$ and $\alpha'$ with domains $X =\dom(\alpha)$ and $X' = \dom(\alpha')$, we define the combination $\alpha \assunion \alpha' : X \cup X' \rightarrow \bool$ as $\alpha \assunion \alpha'(x) = \assignment'(x)$ if $x \in X'$ and $\assignment(x)$ otherwise.
Note that $\assignment'$ overrides $\assignment$ for $x \in X \cap X'$.
We define the \emph{complement} $\overline{\alpha}$ to be $\overline{\alpha}(x) = \neg\alpha(x)$ for all $x \in \dom(\alpha)$.
The complement of a partial assignment is defined analogously with $\neg\bot = \bot$. 
We denote by $\assignment \setminus X$ the assignment without the assignments for every $x \in X$, i.e., $\dom(\assignment \setminus X) = \dom(\assignment) \setminus X$.
The \emph{set of assignments} and \emph{of partial assignments} of $X$ is denoted by $\Assignment(X)$ and $\Assignment_\bot(X)$, respectively.

\begin{example}
  Consider the assignments $\assignment = \set{x \mapsto 0, y \mapsto 1}$ and $\assignment' = \set{x \mapsto 0, y \mapsto 0}$ with $\dom(\assignment) = \dom(\assignment') = \set{x, y}$.
  For the partial assignment $\beta = \set{x \mapsto \bot, y \mapsto 1}$ it holds that $\beta \subassign \assignment$ and $\beta \not\subassign \assignment'$.
  Let $\assignment^* = \assignment \setminus \set{y} = \set{x \mapsto 0}$, then $\assignment \assunion \overline\assignment^* = \set{x \mapsto 1, y \mapsto 1}$.
\end{example}

A quantifier $\quant x \ldot \varphi$ for $\quant \in \set{\exists,\forall}$ \emph{binds} the variable $x$ in the \emph{scope} $\varphi$. 
Variables that are not bound by a quantifier are called \emph{free}.
The set of free variables of formula $\varphi$ is defined as $\free(\varphi)$.
The semantics of the satisfaction relation $\alpha \models \varphi$ is given as
\begin{equation*}
  \begin{array}{ll}
    \alpha \models x &  \text{if } \alpha(x) = 1, \\
    \alpha \models \neg\varphi &  \text{if } \alpha \nmodels \varphi, \\
    \alpha \models \varphi \lor \psi &  \text{if } \alpha \models \varphi \text{ or } \alpha \models \psi,\\
    \alpha \models \varphi \land \psi &  \text{if } \alpha \models \varphi \text{ and } \alpha \models \psi,\\
    \alpha \models \exists x \ldot \varphi &  \text{if some } \alpha' : \set{x} \rightarrow \bool \text{ satisfies } \alpha \assunion \alpha' \models \varphi, \text{ and} \\ 
    \alpha \models \forall x \ldot \varphi &  \text{if all } \alpha' : \set{x} \rightarrow \bool \text{ satisfy } \alpha \assunion \alpha' \models \varphi . \\ 
  \end{array}
\end{equation*}
\emph{QBF satisfiability} is the problem to determine, for a given QBF~$\varphi$, the existence of an assignment $\alpha$ for the free variables of $\varphi$, such that the relation $\models$ holds. 

An existentially quantified variable~$x$ \emph{depends} on all universally quantified variables that are bound prior to $x$.
A universally quantified variable~$x$ \emph{depends} on all existentially quantified variables bound prior to $x$ and additionally on the free variables.
A free variable $x$ \emph{depends} on no variables.
The set of dependencies of $x$ is denoted by $\dep(x)$.
A Boolean function $f : \Assignment(X) \rightarrow \bool$ maps \emph{assignments} of $X$ to true or false.
An assignment $\alpha$ over variables $X$ can be identified by the conjunctive formula $\bigwedge_{x \in X \mid \alpha(x)=1} x \land \bigwedge_{x \in X \mid \alpha(x)=0} \neg x$.
Similarly, Boolean functions can be represented by propositional formulas over the variables in their domain.
Let $\varphi[f_{x_1},\dots,f_{x_n}]$ be the propositional formula where occurrences of $x_i$ are replaced by the propositional representation of $f_{x_i}$.
It is defined as
\begin{align*}
  x[f_{x_1},\dots,f_{x_n}] &{}=
  \begin{cases}
    f_{x_i} & \text{if } x = x_i \text{ for some } i\\
    x & \text{otherwise} 
  \end{cases} \\
  (\neg\varphi)[f_{x_1},\dots,f_{x_n}] &{}= \neg (\varphi[f_{x_1},\dots,f_{x_n}]) \\
  (\varphi \lor \psi)[f_{x_1},\dots,f_{x_n}] &{}= (\varphi[f_{x_1},\dots,f_{x_n}]) \lor (\psi[f_{x_1},\dots,f_{x_n}]) \\
  (\varphi \land \psi)[f_{x_1},\dots,f_{x_n}] &{}= (\varphi[f_{x_1},\dots,f_{x_n}]) \land (\psi[f_{x_1},\dots,f_{x_n}]) \\
  (\exists x\ldot \varphi) [f_{x_1},\dots,f_{x_n}] &{}= \varphi[f_{x_1},\dots,f_{x_n}] \\
  (\forall x\ldot \varphi) [f_{x_1},\dots,f_{x_n}] &{}= \varphi[f_{x_1},\dots,f_{x_n}] 
\end{align*}
For example, let $\varphi = \forall x \ldot \exists y \ldot (x \lor \neg y) \land (\neg x \lor y)$ and let $f_y(x) = x$, then $\varphi[f_y] = (x \lor \neg x) \land (\neg x \lor x)$.
A witness for a satisfiable QBF is a \emph{Skolem function} $f_x :\Assignment(\dep(x)) \rightarrow \bool$ for every variable $x$ that is free or existentially quantified, such that $\neg\varphi[f_{x_1},\dots,f_{x_n}]$ is unsatisfiable.
For unsatisfiable QBFs, the witnesses are defined dually and called \emph{Herbrand functions}.
We use the notation $\varphi[\assignment]$ to replace variables $x \in \dom(\assignment)$ by their assignments~$\assignment(x)$.

A \emph{closed} QBF is a formula without free variables. 
Closed QBFs are either true or false. 
A formula is in prenex form, if the formula consists of a quantifier prefix followed by a propositional formula. 
Every QBF can be transformed into a closed QBF and into prenex form while maintaining satisfiability.
A \emph{literal} $l$ is either a variable $x \in X$, or its negation $\neg x$. 
Given a set of literals $\set{l_1,\dots,l_n}$, the disjunctive combination $(l_1 \lor \ldots \lor l_n)$ is called a \emph{clause} and the conjunctive combination  $(l_1 \land \ldots \land l_n)$ is called a \emph{cube}.
We denote by $\var(l)$ the operation that returns the variable corresponding to $l$.
A QBF is in negation normal form (NNF) if negation is only applied to variables.
Every QBF can be transformed into NNF by at most doubling the size of the formula and without introducing new variables.
For formulas in NNF, we treat literals as atoms.

%===============================================================================
\section{Abstraction-based Algorithm} \label{sec:solving}
%===============================================================================

For QBFs given in CNF, there are recursive refinement algorithms where the refinement is based on clauses~\cite{conf/ijcai/JanotaM15,conf/fmcad/RabeT15,conf/cav/Tentrup17}.
The underlying insight is that multiple variable assignments may lead to the satisfaction of the same clauses, hence, instead of communicating assignments, the information whether a clause is satisfied or not is communicated between quantifier blocks.
Instead of excluding assignments one at a time, those algorithms may exclude multiple assignments with a single refinement step.
In the following, we propose a generalization to formulas in negation normal form, i.e., we base the communication on the satisfaction of individual subformulas.
For this section, we assume an arbitrary (closed, prenex) QBF $\Phi = \quant X_1 \dots \quant X_n \ldot \varphi$ with quantifier prefix $\quant X_1 \dots \quant X_n$ and propositional body $\varphi$ in NNF.

\smallparagraph{SAT solver.}
We use a generic solving function $\Call{sat}{\theta, \assignment}$ for propositional formula $\theta$ and assignment~$\assignment$, that returns whether $\theta \land \assignment$ is satisfiable.
In the positive case, written $\Call{sat}{\theta, \assignment} \Rightarrow \SAT(\assignment')$, it returns a satisfying assignment $\assignment'$.
We write $\Call{sat}{\theta, \assignment} \Rightarrow \SAT(\assignment_V)$ if we are only interested in a subset $V$ of the variables in $\theta$.
In the negative case, written $\Call{sat}{\theta, \assignment} \Rightarrow \UNSAT(\beta)$, it returns a partial assignment $\beta \subassign \assignment$ such that $\theta \land \beta$ is unsatisfiable.

\begin{example}
  We show a few examples of the usage of the $\sat$ function using $\theta = (x \lor (\overline{x} \land y))$ and $\overline\theta = (\overline{x} \land (x \lor \overline{y}))$.
  \begin{align*}
      &\Call{sat}{\overline\theta, \set{}} \Rightarrow \SAT(\assignment_\set{x}) \hspace{19pt} \text{where } \assignment_\set{x} = \set{x \mapsto 0} \\
      &\Call{sat}{\theta, \assignment_\set{x}} \Rightarrow \SAT(\assignment_\set{y})  \quad \text{where } \assignment_\set{y} = \set{y \mapsto 1} \\
      &\Call{sat}{\overline\theta, \assignment_\set{x} \assunion \assignment_\set{y}} \Rightarrow \UNSAT(\set{x \mapsto 0, y \mapsto 1})
  \end{align*}
\end{example}

\smallparagraph{Notation.}
To facilitate working with arbitrary Boolean formulas, we start with introducing additional notation.
Let $\boolf$ be the set of Boolean formulas and let $\subf(\psi) \subset \boolf$ ($\dsubf(\psi) \subset \boolf$) be the set of (direct) subformulas of $\psi$ (note that $\psi \in \subf(\psi)$ but $\psi \notin \dsubf(\psi)$).
For a propositional formula $\psi$, $\sftype(\psi) \in \set{\lit,\lor,\land}$ returns the Boolean connector if $\psi$ is not a literal.
For example, given $\psi = (x_1 \lor (\overlineindex{x_1} \land x_2))$, $\subf(\psi) = \set{(x_1 \lor (\overlineindex{x_1} \land x_2)), x_1, (\overlineindex{x_1} \land x_2), \overlineindex{x_1}, x_2}$, $\dsubf(\psi) = \set{x_1, (\overlineindex{x_1} \land x_2)}$, and $\sftype(\psi) = \lor$.

\smallparagraph{Interface Variables.}
To communicate the value of subformulas, we introduce two special types of variables which we call \emph{interface variables}.
The value of those variables represents whether the value of a subformula is determined and in the positive case, the value itself.
We only consider existential quantifiers as the definition for the universal quantifiers is dual with respect to negation.
For a quantifier $\exists X$, we say that a subformula $\psi$ is \emph{positive} if $\psi$ is conjunctive ($\sftype(\psi) = \land$) and not falsified or $\psi$ is disjunctive ($\sftype(\psi) = \lor$) and satisfied.
A subformula is \emph{negative} if it is not positive.
At quantifier $\exists X$, variable assignments determine whether a subformula $\psi$ is positive or negative.
The interface variables~$t_\psi$ and $b_\psi$ represent whether $\psi$ is positive, with the difference that $t_\psi$ combines assignments from variables bound by outer quantifiers whereas $b_\psi$ combines assignments from variables bound by outer quantifiers and from variables $X$.
We denote the set of variables $t_\psi$ by $T$ and call them $T$~variables and analogously the set of variables $b_\psi$ by $B$ and call them $B$ variables.
For solving $\quant X$ in the abstraction algorithm, we replace the communication of variable assignments by communicating assignments to $T$ and $B$ variables.

\smallparagraph{Abstractions.}
An \emph{abstraction} for quantifier $\quant X$ is a propositional formula $\theta_X$ over variables $X$, $T$, and $B$.
The sets $T_X$ and $B_X$ contain those $T$ and $B$ variables that are used for communication at this quantifier level.
Unlike previous approaches that utilize SAT solvers~\cite{journals/ai/JanotaKMC16,conf/ijcai/JanotaM15,conf/fmcad/RabeT15,conf/sat/TuHJ15}, we keep for every quantifier level a \emph{dual abstraction} $\overline{\theta}_X$ that is used for optimization of abstraction entries and translating interface variables.
When translating a $B$ variable to a $T$ variable, we use the same index, e.g., in line~\ref{alg:abstraction-translation} of Alg.~\ref{alg:abstraction}, the $B$ variables~$b_{\psi}$ are translated to $T$ variables $t_{\psi}$ of the inner quantifier.
Before going into algorithmic details, we describe an execution of Algorithm~\ref{alg:abstraction} on $\Phi_\mathit{ex}$.
\begin{example} \label{ex:abstraction-algorithm}
Consider the example $\Phi_\mathit{ex} = \forall x \ldot \exists y \ldot \overbrace{(x \lor \underbrace{(\overline{x} \land y)}_{\psi_2})}^{\psi_1}$ and its negation $\neg\Phi_\mathit{ex} = \exists x \ldot \forall y \ldot \overbrace{(\overline{x} \land \underbrace{(x \lor \overline{y})}_{\neg\psi_2})}^{\neg\psi_1}$.
Assume that $\theta_x = b_1\land(b_1 \limplies \overline{x})\land(b_2 \limplies x)$ is the abstraction for quantifier $\forall x$ and that $\theta_y = (t_1 \lor b_2)\land(b_2 \limplies t_2)\land(b_2 \limplies y)$ is the abstraction for quantifier $\exists y$ (the definition of abstraction is given later).

The algorithm starts with the top level quantifier $\forall x$.
The universal player has to set $x$ to false to satisfy $\theta_x$, leading to the unique assignment $\assignment_B = \set{b_1 \mapsto 1, b_2 \mapsto 0}$ of B variables.
This means that $\neg\psi_1$ is positive and $\neg\psi_2$ is negative in $\neg\Phi_\mathit{ex}$.
Due to duality, $\psi_1$ is negative and $\psi_2$ is positive in $\Phi_\mathit{ex}$.
Thus, the assignment $\assignment_B$ is translated into assignment $\assignment_T$ of $T$ variables in $\theta_y$ by negation ($\assignment_T = \set{ t_1 \mapsto 0, t_2 \mapsto 1}$, line~\ref{alg:abstraction-translation}).
At the existential quantifier $\exists y$, the abstraction $\theta_y$ is solved under the assumption $\assignment_T$, resulting in a satisfiable query with variable assignment $\set{y \mapsto 1}$.
We then use the dual abstraction $\overline{\theta}_y = b_1\land( b_1 \limplies t_1)\land(b_1 \limplies b_2)\land(b_2 \limplies t_2 \lor \overline{y})$ to translate $\assignment_T$ to a partial assignment $\beta_T$ (line~\ref{alg:abstraction-return-propositional}).
The partial assignment $\set{t_2 \mapsto 0, y \mapsto 1}$ is enough to falsify the dual abstraction $\overline{\theta}_y$, thus, the assumption $\assignment_T(t_2) = 1$ is needed to satisfy $\theta_y$ and the partial assignment $\beta_T$ with $\beta(t_2) = 1$ is returned (line~\ref{alg:abstraction-return-propositional}).
The following refinement forces that $b_2$ must be set to true in the next iteration, i.e., $\theta_x = \theta_x \land b_2$.
This depletes all possible assignments of the universal quantifier, thus, proving that the instance is true.
\end{example}

\begin{algorithm}[t]
\caption{Abstraction Based Algorithm} \label{alg:abstraction}
\begin{algorithmic}[1]
\Procedure{abstraction-qbf-rec}{$\quant X \ldot \varphi, \assignment_{T_X}$}
  \While{$\Call{sat}{\theta_X, \assignment_{T_X}} \Rightarrow \SAT(\assignment_X \assunion \assignment_{B_X})$} \label{alg:abstraction-sat-query} \Comment{generate candidate $\assignment_{B_X}$}
    \If{$\varphi$ is propositional} \label{alg:abstraction-start-verification}
      \State \Return $\tuple{\SAT_Q, \Call{dual-opt}{\assignment_X, \assignment_{T_X}}}$ \label{alg:abstraction-return-propositional}
    \EndIf
    \State $\assignment_{T_Y} \gets \Call{translate}{X, \varphi, \assignment_{B_X}}$
    \State $\tuple{\res, \beta_{T_Y}} \gets \Call{abstraction-qbf-rec}{\varphi, \assignment_{T_Y}}$ \label{alg:abstraction-recursive-call} \Comment{verify $\assignment_{T_Y}$ recursively }
    \If{$\res = \SAT_Q$}
      \State $\overline{\theta}_X \gets \overline{\theta}_X \land \Call{refine$_X$}{\beta_{T_Y}}$ \label{alg:abstraction-refinement-dual} \Comment{refine dual abstraction}
      \State \Return $\tuple{\SAT_Q, \Call{dual-opt}{\assignment_X, \assignment_{T_X}}}$ \label{alg:abstraction-return-recursion}
    \EndIf
    \State $\theta_X \gets \theta_X \land \Call{refine$_X$}{\beta_{T_Y}}$ \label{alg:abstraction-refinement} \Comment{refine abstraction}
  \EndWhile
  \State let $\beta_{T_X}$ be the failed assumptions ($\Call{sat}{\theta_X, \assignment_{T_X}} \Rightarrow \UNSAT(\beta_{T_X})$)
  \State \Return $\tuple{\UNSAT_Q, \beta_{T_X}}$ \Comment{$\beta_{T_X} \subassign \assignment_{T_X}$}
\EndProcedure
\Procedure{refine$_X$}{$\beta_{T_Y}$}
  \State \Return $\bigvee_{b \in B} b$ \textbf{where} $B = \set{b_{\psi} \in B_X \mid \beta_{T_Y}(t_{\psi})=1}$
\EndProcedure
\Procedure{dual-opt}{$\assignment_X$, $\assignment_{T_X}$}
  \State $\Call{sat}{\overline{\theta}_X, \assignment_X \assunion \overline{\assignment}_{T_X}} \Rightarrow \UNSAT(\beta_{T_X})$ \label{alg:abstraction-dual-opt}
  \State \Return $\overline{\beta}_{T_X}$ \Comment{$\overline{\beta}_{T_X} \subassign \assignment_{T_X}$}
\EndProcedure
  \Procedure{translate}{$X$, $\quant\, Y \ldot \varphi$, $\assignment_{B_X}$}
    \State \Return $\assignment_{T_Y}$ s.t. $\assignment_{T_Y}(t_{\psi}) = \overline{\assignment}_{B_X}(b_{\psi})$ for all $t_{\psi} \in T_Y$ \label{alg:abstraction-translation} \Comment{translate $\assignment_{B_X} \rightarrow \assignment_{T_Y}$}
  \EndProcedure
  \Procedure{abstraction-qbf}{$\quant X_1 \dots \quant X_n \ldot \varphi$}
    \State \textbf{for all} $\quant X_i$, initialize $\theta_{X_i}$ and $\overline{\theta}_{X_i}$
    \State \Return $\Call{abstraction-qbf-rec}{\quant X_1 \dots \quant X_n \ldot \varphi, \set{}}$
  \EndProcedure
\end{algorithmic}
\end{algorithm}

\smallparagraph{Algorithm.}
The main procedure of Algorithm~\ref{alg:abstraction} is \textsc{abstraction-qbf-rec}, that recurses on the quantifier prefix.
In line~\ref{alg:abstraction-sat-query}, a candidate solution is generated (represented by an assignment $\assignment_{B_X}$) with respect to an assignment $\assignment_{T_X}$ given by the outer quantifier.
In the following, the candidate solution is verified recursively (line~\ref{alg:abstraction-recursive-call}) and in the negative case the abstraction $\theta_X$ is refined by a blocking clause (line~\ref{alg:abstraction-refinement}) that eliminates (at least) this candidate.

To verify a candidate $\assignment_{B_X}$ recursively, it is \emph{translated} into an assignment $\assignment_{T_Y}$ of the inner quantifier $\overline{\quant}\, Y$ (line~\ref{alg:abstraction-translation}).
This is done by negation since a subformula $\psi$ is positive for $\quant X$ iff it is negative for $\overline{Q}\,Y$.
The \emph{refinement} operation generates, given a counterexample represented as a partial assignment $\beta_{T_Y}$, a clause consisting of $B$ variables that excludes this counterexample:
For every $T$ variable $t_\psi$ that is contained in the counterexample and is positive for the inner quantifier, the refinement adds a $B$ variable $b_\psi$ meaning that one of those subformulas must be positive for $\quant X$ in the next iteration.

The dual abstraction $\overline{\theta}_X$ is defined as the abstraction for $\overline{\quant} X$ and is used in two ways.
First, it optimizes that candidate $\assignment_{B_X}$ in the propositional case (line~\ref{alg:abstraction-return-propositional}), i.e., it generates potentially smaller witnesses.
Second, it translates an assignment of $T_Y$ variables $\beta_{T_Y}$ to an assignment of $T_X$ variables $\beta_{T_X}$ that is returned to the outer quantifier (line~\ref{alg:abstraction-return-recursion}).

We now focus on the abstraction $\theta_X$.
In Example~\ref{ex:abstraction-algorithm}, we have already seen an instance of the abstraction that we formally introduce in the following.
The abstraction is a modification of the Plaisted-Greenbaum encoding~\cite{journals/jsc/PlaistedG86}: for subformula $\psi$, the $b$-literal $b_\psi$ corresponds to the defining literal of the Plaisted-Greenbaum encoding (see definition of $\enc$ in Fig.~\ref{fig:abstraction}).
$\enc_\psi(\psi')$ is responsible for abstracting from actual assignments: literals bound at the quantifier are returned unchanged, literals bound at an outer quantifier are abstracted as a $T$ variable, and we use the defining literal $b_{\psi'}$ of a subformula $\psi'$ if the valuation of the subformula is guaranteed to be fixed, i.e., there is no inner influence.

Given a propositional formula $\varphi$ in NNF and a quantifier $\exists X$, we build the following propositional formula in CNF representing the abstraction 
$\theta_X = \encout_\varphi(\varphi) \land \bigwedge_{\psi \in \subf(\varphi) \land \sftype(\psi) \neq \lit} \enc(\psi) $
for this quantifier, where $\encout$ encodes that $\varphi$ must hold and $\enc$ defines a CNF formula that encodes the truth of subformula $\psi$ with respect to the valuations of the current, inner, and outer quantifier represented by $B$ and $T$ variables, respectively.
The definitions are given in Fig.~\ref{fig:abstraction}.
\begin{figure}[t]
  \begin{align*}
    \enc(\psi) &{}=
      \begin{cases}
        \displaystyle\bigwedge\limits_{\substack{\psi' \in \dsubf(\psi)\\\enc_\psi(\psi')  \neq \bot}} \left(b_\psi \limplies \enc_\psi(\psi') \right) & \text{if } \sftype(\psi) = \land \\
        b_\psi \limplies \displaystyle\bigvee\limits_{\substack{\psi' \in \dsubf(\psi)\\\enc_\psi(\psi')  \neq \bot}} \enc_\psi(\psi') & \text{if } \sftype(\psi) = \lor 
      \end{cases}\\
    \enc_\psi(\psi') &{}= 
      \begin{cases}
        \psi' & \text{if } \sftype(\psi') = \lit \land \var(\psi') \in X \\
        t_\psi & \text{if } \sftype(\psi') = \lit \land \text{literal is bound by outer quantifier}\\
        b_{\psi'} & \text{if } \sftype(\psi') \neq \lit \land \psi' \text{ has no inner influence}\\
        \bot & \text{otherwise}
      \end{cases}\\
    \encout_\psi(\psi') &{}=
      \begin{cases}
        b_{\psi'} & \text{if } \sftype(\psi') = \land \\
        \bigvee_{\psi^* \in \dsubf(\psi')} \encout_{\psi'}(\psi^*) & \text{if } \sftype(\psi') = \lor \\
        \psi' & \text{if } \sftype(\psi') = \lit \land \var(\psi') \in X \\
        t_\psi & \text{if } \sftype(\psi') = \lit \land \text{literal is bound by outer quantifier}\\
        \neg b_\psi & \text{if } \sftype(\psi') = \lit \land \text{literal is bound by inner quantifier}\\
      \end{cases}
  \end{align*}
  \caption{Definition of the abstraction for quantifier block $\exists X$.}
  \label{fig:abstraction}
\end{figure}
The abstraction of a quantifier $\forall X$ is defined as the existential abstraction for $\neg\varphi$.
Note that not every $B$ literal that is used in the abstraction may be exposed as an interface literal.
For the given abstraction, we define the set of interface variables for quantifier $\quant X_i$ as
\begin{align*}
  &B_{X_i} = \set{b_\psi \mid \psi \in \subf(\varphi) \land \psi \text{ contains a variable bound} \leq i}, \text{ and}\\
  &T_{X_i} = \set{t_\psi \mid \psi \in \subf(\varphi) \land \psi \text{ contains a variable bound} < i}.
\end{align*}

\begin{theorem} \label{thm:abstraction-qbf-correct}
  $\absqbf$ is sound and complete.
\end{theorem}
The proof is given in Section~\ref{sec:correctness} and relies on techniques developed for certification in the next section.

%===============================================================================
\section{Certification} \label{sec:certification}
%===============================================================================

Certification is an essential component of QBF solving.
Certification amounts to extracting witnessing functions from a QBF, either \emph{Skolem} functions for true QBFs or \emph{Herbrand} functions for false QBFs.
Not only does it allow to verify the solver result, but the resulting functions can also be used in the context of the application.
The main result of this section is a proof format for our abstraction algorithm and an efficient algorithm to transform proof traces into Boolean functions.

\smallparagraph{Proof Format.}
To extract a witness from a run of $\absqbf$, we need to remember \emph{situations} and \emph{reactions}, represented by assignments to $T$ and $X$, that were satisfiable for the respective quantifier.
Hence, the proof $\absproof$ consists of a sequence of pairs $\tuple{\beta_T,\assignment_X} \in (\Assignment_\bot(T) \times \Assignment(X))$ and these pairs can be obtained from the algorithm by the result $\beta_{T_X}$ of the query to the dual abstraction~$\overline{\theta}_X$ in line~\ref{alg:abstraction-dual-opt}.
As an immediate consequence, the number of pairs in the proof trace is linear in the number of iterations of the algorithm.
We define a function $\invabs_X : \Assignment_\bot(T) \to \boolf(V)$ which, for a given quantifier $\quant X$, maps an assignment $\beta_T$ to a Boolean formula over variables $V$ bound by outer quantifiers (with respect to $X$).
Intuitively, $\invabs_X(\beta_T)$ describes those assignments that lead to $\beta_T$ in the abstraction of quantifier $\quant X$.

\smallparagraph{Function Extraction.}
Prior to the function extraction, we filter out those pairs from the proof $\absproof$ that correspond to the variables that are dependencies and do not describe a function, i.e., universal variables for true QBFs and existential variables for false QBFs.
The remaining proof consists of pairs $\tuple{\beta_T,\assignment_X}$ where $\invabs_X(\beta_T)$ is a formula that represents the situation where the response $\assignment_X$ is correct.
Let $\tuple{\beta_T^1,\assignment_X^1}\dots\tuple{\beta_T^n,\assignment_X^n}$ be the pairs corresponding to quantifier $\quant X$ and let $x \in X$ be some variable, the function $f_x : \Assignment(\dep(x)) \to \bool$ is defined as
\begin{equation}
  f_x \equiv \bigvee_{i=1}^n \left((\assignment_X^i(x)=1) \land\invabs_X(\beta_T^i) \land \bigwedge_{j < i}\neg \invabs_X(\beta_T^j) \right)
\end{equation}
This construction is similar to previous extraction algorithms, including~\cite{conf/fmcad/RabeT15,conf/innovations/BeyersdorffBC16}.
The definition of $\invabs_X$ allows that $f_x$ may depend on variables in outer quantifiers corresponding to functions instead of dependencies.
By replacing those variables with their extracted functions, one can make sure that $f_x$ depends only on $\dep(x)$.
The size of $f_x$, measured in terms of distinct subformulas, is linear in the number of pairs and, hence, linear in the size of the proof.

\begin{theorem}
  Given a QBF $\Phi$ and proof trace $\absproof$, the runtime of the function extraction algorithm is in $\bigO(\card{\absproof})$.
  The size of the resulting functions is linear in the size of~$\absproof$.
\end{theorem}

\smallparagraph{Certification.}
A \emph{certificate} is a representation of all functions that, combined, witness the result of the QBF.
A certificate is correct, if two conditions are satisfied: the certificate is (1) functionally correct and (2) well-formed.
Functional correctness can be checked by a propositional SAT query to $\neg\varphi$ (respectively $\varphi$ for false QBFs) where every occurrence of a function variable $y$ is replaced by the function $f_y$.
The unsatisfiability of this query witnesses functional correctness.
The well-formedness criterion concerns the representation of the certificate, usually as a circuit, and requires that the representation of a function depends only on its dependencies.
One can further differentiate \emph{syntactical} and \emph{semantical} well-formedness.
A certificate is syntactically ill-formed if a non-dependency is reachable from the output of a function.
A certificate is semantically ill-formed if a valuation change of a set of non-dependencies changes the valuation of a function.
Our function extraction guarantees syntactical well-formedness and therefore any further circuit simplification guarantees at least semantical well-formedness.

\begin{example} \label{ex:certificate}
  Consider again our example $\Phi_\mathit{ex} = \forall x \ldot \exists y \ldot \overbrace{(x \lor \underbrace{(\overline{x} \land y)}_{\psi_2})}^{\psi_1}$.
  It holds that $\invabs_\set{y}(\set{t_1 \mapsto 1}) = x$ and $\invabs_\set{y}(\set{t_2 \mapsto 1}) = \overline{x}$, because setting $x$ to true satisfies $\psi_1$ and setting it to false does not falsify~$\psi_2$.
  The proof trace for $\Phi_\mathit{ex}$ is $\tuple{\set{t_2 \mapsto 1}, \set{y \mapsto 1}}$ (see~Example~\ref{ex:abstraction-algorithm}) and the resulting Skolem function is $f_y(x) \equiv \invabs_\set{y}(\set{t_2 \mapsto 1}) \equiv \overline{x}$.
  To verify $f_y$, we check $\neg\Phi_\mathit{ex}[f_y] = \exists x \ldot (\overline{x} \land (x \lor \overline{\overline{x}}))$ for unsatisfiability.
\end{example}

%-------------------------------------------------------------------------------
\section{Correctness} \label{sec:correctness}
%-------------------------------------------------------------------------------

In the following, we formalize properties of the abstraction and prove the algorithm correct.
To relate variable assignments and assignments of $T$ variables, we use the function $\invabs_X$ which is defined in the previous section.
$\invabs_X(\beta_T)$ is a propositional formula over the outer variables $V$ (with respect to $X$) that describes the assignments leading to $\beta_T$ in the abstraction of quantifier $\quant X$.
An assignment $\assignment_V$ is \emph{compatible} with a partial assignment $\beta_{T_X}$, if it satisfies $\invabs_X(\beta_{T_X})$.
Then, we write $\assignment_V \compat \beta_{T_X}$ for short.

The proof of Theorem~\ref{thm:abstraction-qbf-correct} is done by induction on the structure of the quantifier prefix.
To prove the base case of the induction, Lemma~\ref{thm:abstraction-correctness}.\ref{thm:abstraction-inner-sat} states that a satisfiable result in the innermost quantifier corresponds to satisfaction and falsification of the propositional formula for the existential and universal player, respectively (see line~\ref{alg:abstraction-return-propositional} of Algorithm~\ref{alg:abstraction}).
Further, Lemma~\ref{thm:abstraction-correctness}.\ref{thm:abstraction-initial} states the correctness of early termination, i.e., if the initial abstraction returns unsatisfiable, the propositional formula is unsatisfiable under the current assignment (dual for universal player).
\begin{lemma}
  \label{thm:abstraction-correctness}
  The abstraction has the following properties:
  \begin{enumerate}
    \item For a quantifier alternation $\quant X \ldot \overline{\quant}\, Y$, the set of outer $B$ literals $B_X$ matches the set of inner $T$ literals $T_Y$, i.e., $\set{\psi \mid b_\psi \in B_X} = \set{\psi \mid t_\psi \in T_Y}$. \label{thm:abstraction-correctness-matching-sets}
    \item If $\theta_{X}$ is satisfiable under assumptions $\assignment_{T_{X}}$ where $X$ is the innermost quantifier, then for all assignments $\assignment$ with $\assignment \compat \assignment_{T_X}$, there is an assignment $\assignment^*$ with $\assignment \subassign \assignment^*$ such that $\Phi[\assignment^*]$ is true ($\quant = \exists$), respectively false ($\quant = \forall$). \label{thm:abstraction-inner-sat}
    \item If $\theta_X$ is unsatisfiable under assumptions $\beta_{T_X}$, then for all assignments $\assignment$ with $\assignment \compat \beta_{T_X}$ it holds that $\Phi[\assignment]$ is false if $\quant = \exists$, respectively true if $\quant = \forall$ (dual for $\overline\theta_X$). \label{thm:abstraction-initial}
  \end{enumerate}
\end{lemma}
\begin{proof}
  \begin{enumerate}
    \item Holds by definition of $T_X$ and $B_X$ (Section~\ref{sec:solving}).
    \item The $B$ variables at the innermost level correspond to the auxiliary variables in the encoding due to Plaisted and Greenbaum~\cite{journals/jsc/PlaistedG86}.
      Further, all outer quantified variables are replaced by $T$ variables.
      Both properties together show that the claim holds.
    \item For the innermost abstraction, this claim holds by the same argument as in $(2)$.
      For the other abstractions, note that the formula $\theta_X$ is weaker than the Plaisted-Greenbaum encoding: the encoding $\enc_\psi(\psi')$ of a subformula $\psi$ only takes other subformulas ($\sftype(\psi') \neq \lit$) into account if $\psi'$ is not influenced by a variable bound by an inner quantifier.\qedhere
  \end{enumerate}
\end{proof}

We now have all tools available to prove Theorem~\ref{thm:abstraction-qbf-correct}.
The first two invariants state that Lemma~\ref{thm:abstraction-correctness}.\ref{thm:abstraction-initial} holds during the execution of the algorithm, that is, also after the refinement steps.
The last two invariants connect variable assignments to the result of the recursive call of \textsc{abstraction-qbf-rec}.

\begin{proof}[Proof of Theorem~\ref{thm:abstraction-qbf-correct}]
  \textsc{abstraction-qbf-rec} maintains the following invariants that witness the correctness of \textsc{abstraction-qbf}.
  \begin{enumerate}
    \item If $\theta_X$ is unsatisfiable under assumptions $\beta_{T_X}$, then for all assignments $\assignment$ with $\assignment \compat \beta_{T_X}$ it holds that $\Phi[\assignment]$ is false if $\quant = \exists$, respectively true if $\quant = \forall$. \label{thm:correctness-invariant-abstraction}
    \item If $\overline\theta_X$ is unsatisfiable under assumptions $\beta_{T_X}$, then for all assignments $\assignment$ with $\assignment \compat \beta_{T_X}$ it holds that $\Phi[\assignment]$ is true if $\quant = \exists$, respectively false if $\quant = \forall$. \label{thm:correctness-invariant-dual-abstraction}
    \item If \textsc{abstraction-qbf-rec} returns $\tuple{\SAT,\beta_{T_X}}$, then for all $\assignment$ with $\assignment \compat \beta_{T_X}$ it holds that $\Phi[\assignment]$ is true. \label{thm:correctness-invariant-sat}
    \item If \textsc{abstraction-qbf-rec} returns $\tuple{\UNSAT,\beta_{T_X}}$, then for all $\assignment$ with $\assignment \compat \beta_{T_X}$ it holds that $\Phi[\assignment]$ is false. \label{thm:correctness-invariant-unsat}
  \end{enumerate}
  Claim (\ref{thm:correctness-invariant-abstraction}) and (\ref{thm:correctness-invariant-dual-abstraction}) hold initially by Lemma~\ref{thm:abstraction-correctness}.\ref{thm:abstraction-initial}.
  
  \emph{Base case.} $\Call{abstraction-qbf-rec}{\exists X \ldot \varphi, \assignment_{T_X}}$ where $\varphi$ is propositional (case $\forall$ is dual).
  Assume $\Call{sat}{\theta_X, \assignment_{T_X}}$ is unsatisfiable (line~\ref{alg:abstraction-sat-query}) with failed assumptions $\beta_{T_X}$.
  Then by (\ref{thm:correctness-invariant-abstraction}) for all assignments $\assignment$ with $\assignment \compat \beta_{T_X}$, $\Phi[\assignment]$ is false, proving (\ref{thm:correctness-invariant-unsat}).
  Assume $\Call{sat}{\theta_X, \assignment_{T_X}}$ is satisfiable (line~\ref{alg:abstraction-sat-query}), then by Lemma~\ref{thm:abstraction-correctness}.\ref{thm:abstraction-inner-sat} for all $\assignment^* \compat \assignment_{T_X}$, there is an assignment $\assignment$ with $\assignment \subassign \assignment^*$ such that $\Phi[\assignment]$ is true.
  Together with Lemma~\ref{thm:duality-optimization-abstraction} this proves (\ref{thm:correctness-invariant-sat}).
  
  \emph{Induction step.}
  $\Call{abstraction-qbf-rec}{\exists X \ldot \forall Y \dots \quant X_n \ldot \varphi, \assignment_{T_X}}$ (case $\forall$ is dual).
  Assume that $\Call{sat}{\theta_X, \assignment_{T_X}}$ is unsatisfiable (line~\ref{alg:abstraction-sat-query}) with failed assumptions $\beta_{T_X}$.
  Then by (\ref{thm:correctness-invariant-abstraction}) for all assignments $\assignment^*$ with $\assignment^* \compat \beta_{T_X}$, $\Phi[\assignment^*]$ is false, proving (\ref{thm:correctness-invariant-unsat}).
  Assume $\Call{sat}{\theta_X, \assignment_{T_X}}$ is satisfiable (line~\ref{alg:abstraction-sat-query}).
  The candidate $\assignment_{B_X}$ is translated into an assignment $\assignment_{T_Y}$ (Lemma~\ref{thm:abstraction-correctness}.\ref{thm:abstraction-correctness-matching-sets}).
  The following recursive call (line~\ref{alg:abstraction-recursive-call}) returns either $\SAT$ or $\UNSAT$.
  If the result is $\tuple{\SAT,\beta_{T_Y}}$, then by IH and claim (\ref{thm:correctness-invariant-sat}), for all $\assignment^*$ with $\assignment^* \compat \beta_{T_Y}$ it holds that $\Phi[\assignment^*]$ is true.
  Excluding $\beta_{T_Y}$ from $\overline\theta_X$ (line~\ref{alg:abstraction-refinement-dual}) thus preserves invariant (\ref{thm:correctness-invariant-dual-abstraction}).
  Using invariant (\ref{thm:correctness-invariant-dual-abstraction}) and Lemma~\ref{thm:duality-optimization-abstraction}, this proves (\ref{thm:correctness-invariant-sat}) when returning from line~\ref{alg:abstraction-return-recursion}.
  If the result is $\tuple{\UNSAT,\beta_{T_Y}}$, then by IH and claim (\ref{thm:correctness-invariant-unsat}), for all $\assignment^*$ with $\assignment^* \compat \beta_{T_Y}$ it holds that $\Phi[\assignment]$ is false.
  Excluding $\beta_{T_Y}$ from $\theta_X$ (line~\ref{alg:abstraction-refinement}) preserves invariant (\ref{thm:correctness-invariant-abstraction}).
  Completeness (the while loop cannot execute infinitely often) follows from the fact that there are only finitely many different blocking clauses.
\end{proof}

\begin{lemma} \label{thm:duality-optimization-abstraction}
  Let $\theta_X$ and let $\assignment_{T_X}$ be given.
  If $\assignment_X$ is a satisfying assignment of $\theta_X[\assignment_{T_X}]$, then $\Call{sat}{\overline\theta_X, \assignment_X \assunion \overline\assignment_{T_X}}$ returns $\UNSAT(\beta_{T_X})$ and for all $\assignment^*$ with $\overline\beta_{T_X} \subassign \assignment^*$, $\theta_X[\assignment^* \assunion \assignment_X]$ is true.
\end{lemma}
\begin{proof}
  Note that by definition of $\theta_X$ and $\overline\theta_X$, an assignment $\assignment_{T_X}$ in $\theta_X$ corresponds to an assignment $\overline\assignment_{T_X}$ in the dual abstraction $\overline\theta_X$, i.e, $\assignment_{T_X}$ and $\overline\assignment_{T_X}$ represent the same variable assignments (outer variables w.r.t.~$\quant X$) in $\theta_X$ and $\overline\theta_X$, respectively.
  As $\assignment_X$ is a satisfying assignment for $\theta_X[\assignment_{T_X}]$, $\overline\theta_X[\overline\assignment_{T_X} \assunion \assignment_X]$ is unsatisfiable.
  By definition of failed assumptions, $\beta_{T_X} \subassign \overline\assignment_{T_X}$ and $\Call{sat}{\overline\theta_X, \beta_{T_X}}$ returns $\UNSAT$, i.e., there is no $\assignment$ with $\beta_{T_X} \subassign \assignment$ that satisfies $\overline\theta_X$, hence, all $\assignment^*$ with $\overline\beta_{T_X} \subassign \assignment^*$ satisfy $\theta_X[\assignment_X]$.
\end{proof}

%===============================================================================
\section{Evaluation}
%===============================================================================

We implemented Algorithm~\ref{alg:abstraction} and its optimizations in a solver called $\quabs$%
\footnote{Source code available at \url{https://github.com/ltentrup/quabs}}
(Quantified Abstraction Solver) that takes QBFs in the standard format QCIR.
As the underlying SAT solver, we use CryptoMiniSat~\cite{conf/sat/SoosNC09}.
We compare $\quabs$ against the publicly available QBF solvers that support the QCIR format, namely $\ghostq$~\cite{conf/sat/KlieberSGC10}, $\qfun$~\cite{conf/aaai/Janota18}, $\cqesto$~\cite{conf/sat/Janota18}, and $\qute$~\cite{conf/sat/PeitlSS17}.
For our experiments, we used a machine with a $3.6\,\text{GHz}$ quad-core Intel Xeon processor and $32\,\text{GB}$ of memory.
The timeout and memout were set to $10$ minutes and $8\,\text{GB}$, respectively.

$\quabs$ has been independently evaluated in the annual QBF competition, called \emph{QBFEVAL} and the results of the latest evaluation are given in Table~\ref{tbl:qbfeval-2018}.
Notably, $\quabs$ solved most instances, 21 more than the second best solver.
The certification capabilities of $\quabs$ are used in the reactive synthesis tool $\bosy$~\cite{conf/cav/FaymonvilleFT17}, which won the synthesis track in the reactive synthesis competition~(SYNTCOMP) 2016 and 2017~\cite{journals/corr/JacobsBBK0KKLNP16,journals/corr/abs-1711-11439}.

\begin{table}[t]
  \caption{This table shows the number of solved instances within 10 minutes.}\smallskip
  \centering
  \hfill
  \subtable[QBFEVAL'18\label{tbl:qbfeval-2018}]{
  \begin{tabular}{lllll}
    \hline\noalign{\smallskip}
    Solver & Total & Sat & Unsat & Unique \\
    \noalign{\smallskip}\hline\noalign{\smallskip}
    $\quabs$  & 181 & 82 & 99 & 1 \\
    $\cqesto$ & 160 & 75 & 85 & 1 \\
    $\ghostq$ & 157 & 69 & 88 & 0 \\
    $\qfun$   & 139 & 74 & 65 & 5 \\
    $\qute$   & 116 & 42 & 74 & 0 \\
  \end{tabular}
  }\hfill
  \subtable[Petri Game Benchmarks\label{tbl:petri-games}]{
    \begin{tabular}{lllll}
    \hline\noalign{\smallskip}
    Solver & Total & Sat & Unsat & Unique \\
    \noalign{\smallskip}\hline\noalign{\smallskip}
    $\quabs$  & 195 & 123 & 72 & 14  \\
    $\cqesto$ & 189 & 127 & 62 & 11 \\
    $\qfun$   & 141 & 90  & 51 & 0 \\
    $\ghostq$ & 139 & 85  & 54 & 0 \\
    $\qute$   & 100 & 64  & 36 & 0 \\
  \end{tabular}
  }
  \hfill
\end{table}

\smallparagraph{Certification.}
We implemented the certification approach described in Section~\ref{sec:certification}, but instead of generating proof traces, we build the certificates (represented by And-Inverter Graphs) within the solving loop.
This enables building Skolem and Herbrand functions in parallel during solving and minimizes the certification overhead.
In the verification step, we use CryptoMiniSat to solve the functional correctness query.
The size of a certificate is measured as the number of AND gates.

We evaluate the certification capabilities of QuAbS on synthesis benchmark sets that are designed to take advantage of the structural problem definition.
The \emph{petri-games} benchmark set uses the bounded synthesis approach for Petri games~\cite{journals/corr/abs-1711-10637,journals/iandc/FinkbeinerO17}.
The \emph{safety-synt} benchmark set was created from the safety benchmarks of SYNTCOMP 2014~\cite{journals/sttt/JacobsBBEHKPRRS17}.
The \emph{bounded-synthesis} benchmark set was created from the tool $\bosy$~\cite{conf/cav/FaymonvilleFT17} using the QBF encoding of the reactive synthesis problem using LTL specifications~\cite{conf/tacas/FaymonvilleFRT17}.
The \emph{tree-models} benchmark set was created from LTL benchmarks of SYNTCOMP 2016~\cite{journals/corr/JacobsBBK0KKLNP16}.
All those benchmarks have in common that it is possible to directly build implementations from satisfiable queries.

The overall effect of the certification approach on the runtimes is negligible (less than $1\%$ increase) which we consider as achieving our goal that the combination of solving and certification can be implemented efficiently.
Table~\ref{tbl:number-of-certified-instances-per-family} shows the results of the certification run.
The number of verified instances is lower than the number of solved ones because the verifier exceeded the time- and memory-limit on some instances that could be solved within the limits.
To further reduce the size of certificates, one can employ circuit minimization techniques.
Especially compared to CNF certification, these results are very promising and could boost the use of QBF in synthesis applications.

\begin{table}[t]
\caption{Result of the certification run with timeout of 10 minutes for solving and verification, respectively. The average size of the certificate and the accumulated time spend on solving and verification are restricted to verified instances.}\smallskip
\label{tbl:number-of-certified-instances-per-family}
\centering
\begin{tabular}{lrrrrr}
\hline\noalign{\smallskip}
Benchmark set & \#solved & \#verified & avg.\,size &  solving [sec.] & verification [sec.] \\
\noalign{\smallskip}
\hline
\noalign{\smallskip}
petri-games & 136 & 120 & 90,699 & 5389 & 3508 \\
safety-synt & 160 & 144 & 41,125 & 153 & 2569 \\
bounded-synthesis & 339 & 339 & 11,390 & 4552 & 1457 \\
tree-models & 186 & 179 & 49,456 & 4032 & 9528 
\end{tabular}
\end{table}

%===============================================================================
\subsection{Case Study: Petri Games} \label{sec:pg}
%===============================================================================

In this case study, we outline how the certification capabilities of $\quabs$ can be used for the analysis of unrealizable Petri games and for the construction of implementations from winning strategies.
Petri games~\cite{journals/iandc/FinkbeinerO17,conf/cav/FinkbeinerGO15} represent the synthesis problem for distributed, asynchronous systems with causal memory.
The QBF encoding~\cite{conf/birthday/Finkbeiner15,journals/corr/abs-1711-10637} of those games is particularly challenging for CNF solvers: hardly any instance can be solved, even with enabled preprocessing and independent of the used solver, ruling out existing CNF certification approaches.
In contrast, non-CNF solvers scale much better as shown in Table~\ref{tbl:petri-games}, with $\quabs$ performing best overall.

\smallparagraph{Distributed Synthesis of Asynchronous Systems.}
The manual implementation of programs is a tedious and error-prone task.
The automatic synthesis of a correct implementation for a given specification can help the developer to focus on what requirements to fulfill instead of how to fulfill them.
The intricate communication of asynchronous processes in distributed systems would greatly benefit from the automatic synthesis of correct implementations for each process.
\emph{Petri games} define the synthesis problem of asynchronous, distributed systems with causal memory.
The system is \emph{distributed} in the sense that its consists of local processes with individual strategies without global controller.
The system is \emph{asynchronous} in the sense that local processes advance at individual pace and no global clock exists at which processes produce outputs.
Local strategies at a process can utilize \emph{causal memory} which only allows processes to exchange information upon synchronization.
Petri games are based on an underlying Petri net which makes it possible to utilize the unfolding as representation of causal memory.
The simplest winning condition for Petri games are bad places which the system has to avoid while the environment tries to reach such places.

Consider the example Petri game from Fig.~\ref{fig:pg-no-comm} where the system and the environment can both decide between left and right transitions and the bad place can only be avoided by opposite decisions.
Petri games are an extension of Petri nets where the places are distributed to either belong to the system (gray places) or to the environment (white places).
The tokens flowing through the underlying net now represent players depending on the type of place they are residing in: strategies of system players can restrict which outgoing transitions are allowed to fire whereas environment players decide the flow of tokens in the net.
In the game of Fig.~\ref{fig:pg-no-comm}, the choice of system and environment are \emph{independent}, i.e., the system player has no strategy to avoid the bad place: choosing either the left ($t_\mathit{sl}$) or right ($t_\mathit{sr}$) transition, the environment will do the same, leading the game to the bad place.

\begin{figure}[t]
\centering
    \subfigure[
    A Petri game where the system should not mimic the environment's behavior but no communication takes place prior to the system's decisions.
    \label{fig:pg-no-comm}]{
    \begin{tikzpicture}[node distance=1.25cm and 1cm, >=stealth', bend angle=45, auto, scale=0.8, on grid]
		\node [envplace, tokens=1] 											(e1)	[label=below:$e$]					{};
		\node [envplace, below of=e1, below of=e1, left of=e1] 				(el)  	[label=left:$\mathit{el}$]		    {};
		\node [envplace, below of=e1, below of=e1, right of=e1] 				(er)  	[label=right:$\mathit{er}$]		    {};
		\node [sysplace, below of=e1, below of=e1, tokens=1] 				(s1)  	[label=below:$s$]		     		{};
		\node [sysplace, below of=s1, below of=s1]at ($(el)!0.5!(s1)$)		(sl)  	[label=left:$\mathit{sl}$]		    {};
		\node [sysplace, below of=s1, below of=s1]at ($(s1)!0.5!(er)$)		(sr)  	[label=right:$\mathit{sr}$]		    {};
		\node [sysplace, below of=s1, below of=s1, below of=s1, accepting] 	(bad)  	[label=below:$\mathit{bad}$]		{};
		
		\node [transition, above of=sl, above of=sl, above of=sl] 			(t1) 	[label=left:$t_\mathit{el}$]		{};
		\node [transition, above of=sr, above of=sr, above of=sr] 			(t2) 	[label=right:$t_\mathit{er}$]		{};
		\node [transition, above of=sl] 										(t3) 	[label=left:$t_\mathit{sl}$]		{};
		\node [transition, above of=sr] 										(t4) 	[label=right:$t_\mathit{sr}$]		{};
		\node [transition, below of=sl, left of=sl] 							(t5) 	[label=below:$t_\mathit{bad1}$]		{};
		\node [transition, below of=sr, right of=sr] 						(t6) 	[label=below:$t_\mathit{bad2}$]		{};
		
		\path[-latex, thick]
		 	(t1) 	edge [pre]                            (e1)
					edge [post]                           (el)
					%edge [post]                           (s1)
		 	(t2) 	edge [pre]                            (e1)
					edge [post]                           (er)
					%edge [post]                           (s1)
			(t3) 	edge [pre]                            (s1)
					edge [post]                           (sl)
			(t4) 	edge [pre]                            (s1)
					edge [post]                           (sr)
			(t5) 	edge [pre]                            (el)
					edge [pre]                            (sl)
					edge [post]                           (bad)
			(t6) 	edge [pre]                            (er)
					edge [pre]                            (sr)
					edge [post]                           (bad)
		;
	\end{tikzpicture}
    }
\hspace{2mm}%
    \subfigure[
    The environment forwards its decision to the system and afterwards the system should not mimic this decision.
    \label{fig:pg-sync}]{
    \begin{tikzpicture}[node distance=1.25cm and 1cm, >=stealth', bend angle=45, auto, scale=0.8, on grid]
		\node [envplace, tokens=1] 											(e1)	[label=below:$e$]					{};
		\node [envplace, below of=e1, below of=e1, left of=e1] 				(el)  	[label=left:$\mathit{el}$]		    {};
		\node [envplace, below of=e1, below of=e1, right of=e1] 				(er)  	[label=right:$\mathit{er}$]		    {};
		\node [sysplace, below of=e1, below of=e1] 							(s1)  	[label=below:$s$]		     		{};
		\node [sysplace, below of=s1, below of=s1]at ($(el)!0.5!(s1)$)		(sl)  	[label=left:$\mathit{sl}$]		    {};
		\node [sysplace, below of=s1, below of=s1]at ($(s1)!0.5!(er)$)		(sr)  	[label=right:$\mathit{sr}$]		    {};
		\node [sysplace, below of=s1, below of=s1, below of=s1, accepting] 	(bad)  	[label=below:$\mathit{bad}$]		{};
		
		\node [transition, above of=sl, above of=sl, above of=sl] 			(t1) 	[label=left:$t_\mathit{el}$]		{};
		\node [transition, above of=sr, above of=sr, above of=sr] 			(t2) 	[label=right:$t_\mathit{er}$]		{};
		\node [transition, above of=sl] 										(t3) 	[label=left:$t_\mathit{sl}$]		{};
		\node [transition, above of=sr] 										(t4) 	[label=right:$t_\mathit{sr}$]		{};
		\node [transition, below of=sl, left of=sl] 							(t5) 	[label=below:$t_\mathit{bad1}$]		{};
		\node [transition, below of=sr, right of=sr] 						(t6) 	[label=below:$t_\mathit{bad2}$]		{};
		
		\path[-latex, thick]
		 	(t1) 	edge [pre]                            (e1)
					edge [post]                           (el)
					edge [post]                           (s1)
		 	(t2) 	edge [pre]                            (e1)
					edge [post]                           (er)
					edge [post]                           (s1)
			(t3) 	edge [pre]                            (s1)
					edge [post]                           (sl)
			(t4) 	edge [pre]                            (s1)
					edge [post]                           (sr)
			(t5) 	edge [pre]                            (el)
					edge [pre]                            (sl)
					edge [post]                           (bad)
			(t6) 	edge [pre]                            (er)
					edge [pre]                            (sr)
					edge [post]                           (bad)
		;
	\end{tikzpicture}
    }
\hspace{2mm}%
    \subfigure[
    A winning strategy where the system does not mimic the environment such that transitions to the bad place (dashed in blue) become unreachable.
    \label{fig:pg-strat}]{
    \begin{tikzpicture}[node distance=1.25cm and 1cm, >=stealth', bend angle=45, auto, scale=0.8, on grid]
		\node [envplace, tokens=1] 											(e1)	[label=below:$e$]					{};
		\node [sysplace, below of=e1, below of=e1, opacity=0] 				(s1)  	[]						     		{};
		\node [sysplace, below of=s1, below of=s1]at ($(el)!0.5!(s1)$)		(sl)  	[label=left:$\mathit{sl}$]		    {};
		\node [sysplace, below of=s1, below of=s1]at ($(s1)!0.5!(er)$)		(sr)  	[label=right:$\mathit{sr}'$]		{};
		\node [envplace, above of=sl, above of=sl, left of=sl] 				(el)  	[label=above:$\mathit{el}$]		    {};
		\node [envplace, above of=sr, above of=sr, right of=sr] 				(er)  	[label=above:$\mathit{er}$]		    {};
		\node [sysplace, right of=el]						 				(s1l)  	[label=left:$s$]		     		{};
		\node [sysplace, left of=er] 										(s1r)  	[label=right:$s'$]		     		{};
		\node [sysplace, below of=s1, below of=s1, below of=s1, accepting] 	(bad)  	[label=below:$\mathit{bad}$]		{};
		
		\node [transition, above of=sl, above of=sl, above of=sl] 			(t1) 	[label=left:$t_\mathit{el}$]		{};
		\node [transition, above of=sr, above of=sr, above of=sr] 			(t2) 	[label=right:$t_\mathit{er}$]		{};
		\node [transition, above of=sl] 										(t3) 	[label=left:$t_\mathit{sl}$]		{};
		\node [transition, above of=sr] 										(t4) 	[label=right:$t'_\mathit{sr}$]		{};
		\node [transition, below of=sl, left of=sl] 							(t5) 	[label=below:$t_\mathit{bad1}$]		{};
		\node [transition, below of=sr, right of=sr] 						(t6) 	[label=below:$t'_\mathit{bad2}$]	{};
		%\node [transition, below of=sl, left of=sl, opacity=0] 				(t5) 	[label=below:$\phantom{t_\mathit{bad1}}$]		{};
		%\node [transition, below of=sr, right of=sr, opacity=0] 				(t6) 	[label=below:$\phantom{t_\mathit{bad2}}$]		{};
		
		\path[-latex, thick]
		 	(t1) 	edge [pre]                            (e1)
					edge [post]                           (el)
					edge [post]                           (s1l)
		 	(t2) 	edge [pre]                            (e1)
					edge [post]                           (er)
					edge [post]                           (s1r)
			(t3) 	edge [pre]                            (s1r)
					edge [post]                           (sl)
			(t4) 	edge [pre]                            (s1l)
					edge [post]                           (sr)
		;
		
		\path[-latex, thick, dashed, blue]
			(t5) 	edge [pre]                            (el)
					edge [pre]                            (sl)
					edge [post]                           (bad)
			(t6) 	edge [pre]                            (er)
					edge [pre]                            (sr)
					edge [post]                           (bad)
		;
	\end{tikzpicture}
    }
\caption{An example workflow of designing a Petri game is outlined. $\quabs$ produces counterexamples to any strategy in the left Petri game. From there, it becomes clear that there is no information exchange between the system and the environment. Therefore, the design of the Petri game is changed to the one in the middle where the environment leaks its decision to the system. For this game, we can extract the winning strategy on the right using $\quabs$ which avoids the bad place as the system answers with opposite decisions to the decisions of the environment.}
\label{fig:petri-game}
\end{figure}

\smallparagraph{Strategy Construction and Strategy Refutation.}
As the Petri game in Fig.~\ref{fig:pg-no-comm} has no winning strategy, the QBF encoding~\cite{conf/birthday/Finkbeiner15} is unsatisfiable and $\quabs$ returns a certificate for the universal player.
This certificate represents a flow of tokens leading to the bad state for \emph{every} system strategy. 
When the system only decides to enable $t_\mathit{sl}$ and to not enable $t_\mathit{sr}$ then one counterexample moves the environment token from $e$ to $\mathit{el}$, the system token from $s$ to $\mathit{sl}$, and afterwards fires the transition to reach the bad place.
An analog counterexample is returned when the system enables $t_\mathit{sr}$ and does not enable $t_\mathit{sl}$.
When the system enables neither transition then the counterexample moves the environment token from $e$ to $\mathit{el}$ and then reaches a deadlock without termination.
This situation is forbidden for strategies as otherwise the winning condition of avoiding bad places would be a trivial.
When the system activates both transitions then already the initial marking constitutes a counterexample as the system's decision is non-deterministic.

From these counterexamples, we can derive that we have to introduce communication between the system and the environment.
The easiest way to do so is given in Fig.~\ref{fig:pg-sync} where the system player is created with the decision of the environment and then can only afterwards react to it.
The different causal memory of the system player in $s$ depending on whether $t_\mathit{el}$ or $t_\mathit{er}$ was fired results in the unfolding of $s$ (indicated by $'$), as depicted in Fig.~\ref{fig:pg-strat}.
Then, a winning strategy exists where the system player makes a different decision to the previous environment decision.
The satisfying assignment of $\quabs$ in the QBF encoding of this problem allows to directly remove not activated transitions ($t_\mathit{sr}$ and $t'_\mathit{sl}$) and their resulting unreachable parts of the game, making all transitions to the bad place unreachable (indicated as dashed blue lines in Fig.~\ref{fig:pg-strat}).

%===============================================================================
\section{Related Work}
%===============================================================================

Other QBF solving techniques that use structural information are conceptually very different, such as DPLL like~\cite{journals/constraints/EglySW09,conf/sat/GoultiaevaIB09,conf/sat/KlieberSGC10,conf/sat/PeitlSS17} and expansion~\cite{conf/sat/LonsingB08,conf/date/PigorschS09,journals/ai/JanotaKMC16,conf/aaai/Janota18}.
We extend work on QBF solving techniques that communicate the satisfaction of clauses through a recursive refinement algorithm~\cite{conf/ijcai/JanotaM15,conf/fmcad/RabeT15,conf/cav/Tentrup17} that were limited to conjunctive normal form.
Further, the maintenance of a dual abstraction for optimization is new in this context and the certification approach is different and, as shown in the evaluation, much more efficient than the one presented for $\caqe$~\cite{conf/fmcad/RabeT15}.
The structure of independent quantifiers in non-prenex formulas can be used for parallelization during solving for this kind of algorithms~\cite{conf/sat/Tentrup16}.
$\cqesto$~\cite{conf/sat/Janota18} is a recently introduced circuit solver based on a similar algorithm as presented in this paper.
The algorithm, however, differs in the way abstractions are built: we produce a ``static'' abstraction upfront and learn subformula valuations during solving, while $\cqesto$ evaluates the circuit under the current variable assignments and re-encodes the resulting partial circuit using the Tseitin transformation in each refinement step.
To our knowledge, $\cqesto$ cannot produce certificates.
Certification has been considered in the context of CNF solving techniques~\cite{journals/fmsd/BalabanovJ12,conf/sat/NiemetzPLSB12,conf/fmcad/HeuleSB14,conf/fmcad/RabeT15} but we are not aware of another work considering certification in the more general setting.
The duality of circuit based QBF solving has been used to enhance search based CNF solvers~\cite{conf/aaai/GoultiaevaB10,conf/date/GoultiaevaSB13} but this is different to our use of a dual abstraction during solving.

%===============================================================================
\section{Conclusion}
%===============================================================================

We presented a QBF solving algorithm that exploits the structure in the propositional formula.
Further, we defined a certification format suitable for this algorithm and described an efficient algorithm to extract solution witnesses from true, respectively false, QBFs.
We have implemented the solving and certification techniques in a tool called $\quabs$ which won the QBF competition QBFEVAL'18.
We have achieved our goal of the certification approach having nearly no overhead over pure solving approaches.
For the case study of Petri games, we outlined how the certification techniques of $\quabs$ allow the analysis of unrealizable Petri games and the construction of implementations for realizable Petri games.

\subsection*{Acknowledgments}
We thank Mikol{\'{a}}s Janota for reporting a problem with an earlier formulation of the abstraction and the anonymous reviewers for their helpful comments.

\bibliographystyle{eptcs}
\bibliography{main}

\begin{thebibliography}{10}
\providecommand{\bibitemdeclare}[2]{}
\providecommand{\surnamestart}{}
\providecommand{\surnameend}{}
\providecommand{\urlprefix}{Available at }
\providecommand{\url}[1]{\texttt{#1}}
\providecommand{\href}[2]{\texttt{#2}}
\providecommand{\urlalt}[2]{\href{#1}{#2}}
\providecommand{\doi}[1]{doi:\urlalt{http://dx.doi.org/#1}{#1}}
\providecommand{\bibinfo}[2]{#2}

\bibitemdeclare{article}{journals/fmsd/BalabanovJ12}
\bibitem{journals/fmsd/BalabanovJ12}
\bibinfo{author}{Valeriy \surnamestart Balabanov\surnameend} \&
  \bibinfo{author}{Jie{-}Hong~R. \surnamestart Jiang\surnameend}
  (\bibinfo{year}{2012}): \emph{\bibinfo{title}{Unified {QBF} certification and
  its applications}}.
\newblock {\sl \bibinfo{journal}{Formal Methods in System Design}}
  \bibinfo{volume}{41}(\bibinfo{number}{1}), pp. \bibinfo{pages}{45--65},
  \doi{10.1007/s10703-012-0152-6}.

\bibitemdeclare{inproceedings}{conf/innovations/BeyersdorffBC16}
\bibitem{conf/innovations/BeyersdorffBC16}
\bibinfo{author}{Olaf \surnamestart Beyersdorff\surnameend},
  \bibinfo{author}{Ilario \surnamestart Bonacina\surnameend} \&
  \bibinfo{author}{Leroy \surnamestart Chew\surnameend} (\bibinfo{year}{2016}):
  \emph{\bibinfo{title}{Lower Bounds: From Circuits to {QBF} Proof Systems}}.
\newblock In: {\sl \bibinfo{booktitle}{Proceedings of {ITCS}}},
  \bibinfo{publisher}{{ACM}}, pp. \bibinfo{pages}{249--260},
  \doi{10.1145/2840728.2840740}.

\bibitemdeclare{inproceedings}{conf/vmcai/BloemKS14}
\bibitem{conf/vmcai/BloemKS14}
\bibinfo{author}{Roderick \surnamestart Bloem\surnameend},
  \bibinfo{author}{Robert \surnamestart K{\"{o}}nighofer\surnameend} \&
  \bibinfo{author}{Martina \surnamestart Seidl\surnameend}
  (\bibinfo{year}{2014}): \emph{\bibinfo{title}{{SAT}-Based Synthesis Methods
  for Safety Specs}}.
\newblock In: {\sl \bibinfo{booktitle}{Proceedings of {VMCAI}}}, {\sl
  \bibinfo{series}{LNCS}} \bibinfo{volume}{8318},
  \bibinfo{publisher}{Springer}, pp. \bibinfo{pages}{1--20},
  \doi{10.1007/978-3-642-54013-4_1}.

\bibitemdeclare{article}{journals/constraints/EglySW09}
\bibitem{journals/constraints/EglySW09}
\bibinfo{author}{Uwe \surnamestart Egly\surnameend}, \bibinfo{author}{Martina
  \surnamestart Seidl\surnameend} \& \bibinfo{author}{Stefan \surnamestart
  Woltran\surnameend} (\bibinfo{year}{2009}): \emph{\bibinfo{title}{A solver
  for {QBFs} in negation normal form}}.
\newblock {\sl \bibinfo{journal}{Constraints}}
  \bibinfo{volume}{14}(\bibinfo{number}{1}), pp. \bibinfo{pages}{38--79},
  \doi{10.1007/s10601-008-9055-y}.

\bibitemdeclare{inproceedings}{conf/tacas/FaymonvilleFRT17}
\bibitem{conf/tacas/FaymonvilleFRT17}
\bibinfo{author}{Peter \surnamestart Faymonville\surnameend},
  \bibinfo{author}{Bernd \surnamestart Finkbeiner\surnameend},
  \bibinfo{author}{Markus~N. \surnamestart Rabe\surnameend} \&
  \bibinfo{author}{Leander \surnamestart Tentrup\surnameend}
  (\bibinfo{year}{2017}): \emph{\bibinfo{title}{Encodings of Bounded
  Synthesis}}.
\newblock In: {\sl \bibinfo{booktitle}{Proceedings of {TACAS}}}, {\sl
  \bibinfo{series}{LNCS}} \bibinfo{volume}{10205}, pp.
  \bibinfo{pages}{354--370}, \doi{10.1007/978-3-662-54577-5_20}.

\bibitemdeclare{inproceedings}{conf/cav/FaymonvilleFT17}
\bibitem{conf/cav/FaymonvilleFT17}
\bibinfo{author}{Peter \surnamestart Faymonville\surnameend},
  \bibinfo{author}{Bernd \surnamestart Finkbeiner\surnameend} \&
  \bibinfo{author}{Leander \surnamestart Tentrup\surnameend}
  (\bibinfo{year}{2017}): \emph{\bibinfo{title}{{BoSy}: An Experimentation
  Framework for Bounded Synthesis}}.
\newblock In: {\sl \bibinfo{booktitle}{Proceedings of {CAV}}}, {\sl
  \bibinfo{series}{LNCS}} \bibinfo{volume}{10427},
  \bibinfo{publisher}{Springer}, pp. \bibinfo{pages}{325--332},
  \doi{10.1007/978-3-319-63390-9_17}.

\bibitemdeclare{inproceedings}{conf/birthday/Finkbeiner15}
\bibitem{conf/birthday/Finkbeiner15}
\bibinfo{author}{Bernd \surnamestart Finkbeiner\surnameend}
  (\bibinfo{year}{2015}): \emph{\bibinfo{title}{Bounded Synthesis for {Petri}
  Games}}.
\newblock In: {\sl \bibinfo{booktitle}{Proceedings of Correct System Design}},
  {\sl \bibinfo{series}{LNCS}} \bibinfo{volume}{9360},
  \bibinfo{publisher}{Springer}, pp. \bibinfo{pages}{223--237},
  \doi{10.1007/978-3-319-23506-6_15}.

\bibitemdeclare{inproceedings}{journals/corr/abs-1711-10637}
\bibitem{journals/corr/abs-1711-10637}
\bibinfo{author}{Bernd \surnamestart Finkbeiner\surnameend},
  \bibinfo{author}{Manuel \surnamestart Gieseking\surnameend},
  \bibinfo{author}{Jesko \surnamestart Hecking{-}Harbusch\surnameend} \&
  \bibinfo{author}{Ernst{-}R{\"{u}}diger \surnamestart Olderog\surnameend}
  (\bibinfo{year}{2017}): \emph{\bibinfo{title}{Symbolic vs. Bounded Synthesis
  for Petri Games}}.
\newblock In: {\sl \bibinfo{booktitle}{Proceedings of {SYNT@CAV}}}, {\sl
  \bibinfo{series}{EPTCS}} \bibinfo{volume}{260}, pp. \bibinfo{pages}{23--43},
  \doi{10.4204/EPTCS.260.5}.

\bibitemdeclare{inproceedings}{conf/cav/FinkbeinerGO15}
\bibitem{conf/cav/FinkbeinerGO15}
\bibinfo{author}{Bernd \surnamestart Finkbeiner\surnameend},
  \bibinfo{author}{Manuel \surnamestart Gieseking\surnameend} \&
  \bibinfo{author}{Ernst{-}R{\"{u}}diger \surnamestart Olderog\surnameend}
  (\bibinfo{year}{2015}): \emph{\bibinfo{title}{Adam: Causality-Based Synthesis
  of Distributed Systems}}.
\newblock In: {\sl \bibinfo{booktitle}{Proceedings of {CAV}}}, {\sl
  \bibinfo{series}{LNCS}} \bibinfo{volume}{9206},
  \bibinfo{publisher}{Springer}, pp. \bibinfo{pages}{433--439},
  \doi{10.1007/978-3-319-21690-4_25}.

\bibitemdeclare{article}{journals/iandc/FinkbeinerO17}
\bibitem{journals/iandc/FinkbeinerO17}
\bibinfo{author}{Bernd \surnamestart Finkbeiner\surnameend} \&
  \bibinfo{author}{Ernst{-}R{\"{u}}diger \surnamestart Olderog\surnameend}
  (\bibinfo{year}{2017}): \emph{\bibinfo{title}{Petri games: Synthesis of
  distributed systems with causal memory}}.
\newblock {\sl \bibinfo{journal}{Inf. Comput.}} \bibinfo{volume}{253}, pp.
  \bibinfo{pages}{181--203}, \doi{10.1016/j.ic.2016.07.006}.

\bibitemdeclare{inproceedings}{conf/tacas/FinkbeinerT14}
\bibitem{conf/tacas/FinkbeinerT14}
\bibinfo{author}{Bernd \surnamestart Finkbeiner\surnameend} \&
  \bibinfo{author}{Leander \surnamestart Tentrup\surnameend}
  (\bibinfo{year}{2014}): \emph{\bibinfo{title}{Detecting Unrealizable
  Specifications of Distributed Systems}}.
\newblock In: {\sl \bibinfo{booktitle}{Proceedings of {TACAS}}}, {\sl
  \bibinfo{series}{LNCS}} \bibinfo{volume}{8413},
  \bibinfo{publisher}{Springer}, pp. \bibinfo{pages}{78--92},
  \doi{10.1007/978-3-642-54862-8_6}.

\bibitemdeclare{article}{journals/corr/FinkbeinerT15}
\bibitem{journals/corr/FinkbeinerT15}
\bibinfo{author}{Bernd \surnamestart Finkbeiner\surnameend} \&
  \bibinfo{author}{Leander \surnamestart Tentrup\surnameend}
  (\bibinfo{year}{2015}): \emph{\bibinfo{title}{Detecting Unrealizability of
  Distributed Fault-tolerant Systems}}.
\newblock {\sl \bibinfo{journal}{Logical Methods in Computer Science}}
  \bibinfo{volume}{11}(\bibinfo{number}{3}), \doi{10.2168/LMCS-11(3:12)2015}.

\bibitemdeclare{inproceedings}{conf/aaai/GoultiaevaB10}
\bibitem{conf/aaai/GoultiaevaB10}
\bibinfo{author}{Alexandra \surnamestart Goultiaeva\surnameend} \&
  \bibinfo{author}{Fahiem \surnamestart Bacchus\surnameend}
  (\bibinfo{year}{2010}): \emph{\bibinfo{title}{Exploiting {QBF} Duality on a
  Circuit Representation}}.
\newblock In: {\sl \bibinfo{booktitle}{Proceedings of {AAAI}}},
  \bibinfo{publisher}{{AAAI} Press}.

\bibitemdeclare{inproceedings}{conf/sat/GoultiaevaIB09}
\bibitem{conf/sat/GoultiaevaIB09}
\bibinfo{author}{Alexandra \surnamestart Goultiaeva\surnameend},
  \bibinfo{author}{Vicki \surnamestart Iverson\surnameend} \&
  \bibinfo{author}{Fahiem \surnamestart Bacchus\surnameend}
  (\bibinfo{year}{2009}): \emph{\bibinfo{title}{Beyond {CNF:} {A} Circuit-Based
  {QBF} Solver}}.
\newblock In: {\sl \bibinfo{booktitle}{Proceedings of {SAT}}}, {\sl
  \bibinfo{series}{LNCS}} \bibinfo{volume}{5584},
  \bibinfo{publisher}{Springer}, pp. \bibinfo{pages}{412--426},
  \doi{10.1007/978-3-642-02777-2_38}.

\bibitemdeclare{inproceedings}{conf/date/GoultiaevaSB13}
\bibitem{conf/date/GoultiaevaSB13}
\bibinfo{author}{Alexandra \surnamestart Goultiaeva\surnameend},
  \bibinfo{author}{Martina \surnamestart Seidl\surnameend} \&
  \bibinfo{author}{Armin \surnamestart Biere\surnameend}
  (\bibinfo{year}{2013}): \emph{\bibinfo{title}{Bridging the gap between dual
  propagation and {CNF}-based {QBF} solving}}.
\newblock In: {\sl \bibinfo{booktitle}{Proceedings of {DATE}}},
  \bibinfo{publisher}{IEEE}, pp. \bibinfo{pages}{811--814},
  \doi{10.7873/DATE.2013.172}.

\bibitemdeclare{inproceedings}{conf/fmcad/HeuleSB14}
\bibitem{conf/fmcad/HeuleSB14}
\bibinfo{author}{Marijn \surnamestart Heule\surnameend},
  \bibinfo{author}{Martina \surnamestart Seidl\surnameend} \&
  \bibinfo{author}{Armin \surnamestart Biere\surnameend}
  (\bibinfo{year}{2014}): \emph{\bibinfo{title}{Efficient extraction of
  {Skolem} functions from {QRAT} proofs}}.
\newblock In: {\sl \bibinfo{booktitle}{Proceedings of {FMCAD}}},
  \bibinfo{publisher}{IEEE}, pp. \bibinfo{pages}{107--114},
  \doi{10.1109/FMCAD.2014.6987602}.

\bibitemdeclare{inproceedings}{journals/corr/abs-1711-11439}
\bibitem{journals/corr/abs-1711-11439}
\bibinfo{author}{Swen \surnamestart Jacobs\surnameend},
  \bibinfo{author}{Nicolas \surnamestart Basset\surnameend},
  \bibinfo{author}{Roderick \surnamestart Bloem\surnameend},
  \bibinfo{author}{Romain \surnamestart Brenguier\surnameend},
  \bibinfo{author}{Maximilien \surnamestart Colange\surnameend},
  \bibinfo{author}{Peter \surnamestart Faymonville\surnameend},
  \bibinfo{author}{Bernd \surnamestart Finkbeiner\surnameend},
  \bibinfo{author}{Ayrat \surnamestart Khalimov\surnameend},
  \bibinfo{author}{Felix \surnamestart Klein\surnameend},
  \bibinfo{author}{Thibaud \surnamestart Michaud\surnameend},
  \bibinfo{author}{Guillermo~A. \surnamestart P{\'{e}}rez\surnameend},
  \bibinfo{author}{Jean{-}Fran{\c{c}}ois \surnamestart Raskin\surnameend},
  \bibinfo{author}{Ocan \surnamestart Sankur\surnameend} \&
  \bibinfo{author}{Leander \surnamestart Tentrup\surnameend}
  (\bibinfo{year}{2017}): \emph{\bibinfo{title}{The 4th Reactive Synthesis
  Competition {(SYNTCOMP} 2017): Benchmarks, Participants {\&} Results}}.
\newblock In: {\sl \bibinfo{booktitle}{Proceedings of {SYNT@CAV}}}, {\sl
  \bibinfo{series}{EPTCS}} \bibinfo{volume}{260}, pp.
  \bibinfo{pages}{116--143}, \doi{10.4204/EPTCS.260.10}.

\bibitemdeclare{article}{journals/sttt/JacobsBBEHKPRRS17}
\bibitem{journals/sttt/JacobsBBEHKPRRS17}
\bibinfo{author}{Swen \surnamestart Jacobs\surnameend},
  \bibinfo{author}{Roderick \surnamestart Bloem\surnameend},
  \bibinfo{author}{Romain \surnamestart Brenguier\surnameend},
  \bibinfo{author}{R{\"{u}}diger \surnamestart Ehlers\surnameend},
  \bibinfo{author}{Timotheus \surnamestart Hell\surnameend},
  \bibinfo{author}{Robert \surnamestart K{\"{o}}nighofer\surnameend},
  \bibinfo{author}{Guillermo~A. \surnamestart P{\'{e}}rez\surnameend},
  \bibinfo{author}{Jean{-}Fran{\c{c}}ois \surnamestart Raskin\surnameend},
  \bibinfo{author}{Leonid \surnamestart Ryzhyk\surnameend},
  \bibinfo{author}{Ocan \surnamestart Sankur\surnameend},
  \bibinfo{author}{Martina \surnamestart Seidl\surnameend},
  \bibinfo{author}{Leander \surnamestart Tentrup\surnameend} \&
  \bibinfo{author}{Adam \surnamestart Walker\surnameend}
  (\bibinfo{year}{2017}): \emph{\bibinfo{title}{The first reactive synthesis
  competition {(SYNTCOMP} 2014)}}.
\newblock {\sl \bibinfo{journal}{STTT}}
  \bibinfo{volume}{19}(\bibinfo{number}{3}), pp. \bibinfo{pages}{367--390},
  \doi{10.1007/s10009-016-0416-3}.

\bibitemdeclare{inproceedings}{journals/corr/JacobsBBK0KKLNP16}
\bibitem{journals/corr/JacobsBBK0KKLNP16}
\bibinfo{author}{Swen \surnamestart Jacobs\surnameend},
  \bibinfo{author}{Roderick \surnamestart Bloem\surnameend},
  \bibinfo{author}{Romain \surnamestart Brenguier\surnameend},
  \bibinfo{author}{Ayrat \surnamestart Khalimov\surnameend},
  \bibinfo{author}{Felix \surnamestart Klein\surnameend},
  \bibinfo{author}{Robert \surnamestart K{\"{o}}nighofer\surnameend},
  \bibinfo{author}{Jens \surnamestart Kreber\surnameend},
  \bibinfo{author}{Alexander \surnamestart Legg\surnameend},
  \bibinfo{author}{Nina \surnamestart Narodytska\surnameend},
  \bibinfo{author}{Guillermo~A. \surnamestart P{\'{e}}rez\surnameend},
  \bibinfo{author}{Jean{-}Fran{\c{c}}ois \surnamestart Raskin\surnameend},
  \bibinfo{author}{Leonid \surnamestart Ryzhyk\surnameend},
  \bibinfo{author}{Ocan \surnamestart Sankur\surnameend},
  \bibinfo{author}{Martina \surnamestart Seidl\surnameend},
  \bibinfo{author}{Leander \surnamestart Tentrup\surnameend} \&
  \bibinfo{author}{Adam \surnamestart Walker\surnameend}
  (\bibinfo{year}{2016}): \emph{\bibinfo{title}{The 3rd Reactive Synthesis
  Competition {(SYNTCOMP} 2016): Benchmarks, Participants {\&} Results}}.
\newblock In: {\sl \bibinfo{booktitle}{Proceedings of {SYNT@CAV}}}, {\sl
  \bibinfo{series}{EPTCS}} \bibinfo{volume}{229}, pp.
  \bibinfo{pages}{149--177}, \doi{10.4204/EPTCS.229.12}.

\bibitemdeclare{inproceedings}{conf/sat/Janota18}
\bibitem{conf/sat/Janota18}
\bibinfo{author}{Mikol{\'{a}}s \surnamestart Janota\surnameend}
  (\bibinfo{year}{2018}): \emph{\bibinfo{title}{Circuit-Based Search Space
  Pruning in {QBF}}}.
\newblock In: {\sl \bibinfo{booktitle}{Proceedings of {SAT}}}, {\sl
  \bibinfo{series}{LNCS}} \bibinfo{volume}{10929},
  \bibinfo{publisher}{Springer}, pp. \bibinfo{pages}{187--198},
  \doi{10.1007/978-3-319-94144-8\_12}.

\bibitemdeclare{inproceedings}{conf/aaai/Janota18}
\bibitem{conf/aaai/Janota18}
\bibinfo{author}{Mikol{\'{a}}s \surnamestart Janota\surnameend}
  (\bibinfo{year}{2018}): \emph{\bibinfo{title}{Towards Generalization in {QBF}
  Solving via Machine Learning}}.
\newblock In: {\sl \bibinfo{booktitle}{Proceedings of {AAAI}}},
  \bibinfo{publisher}{{AAAI} Press}.

\bibitemdeclare{article}{journals/ai/JanotaKMC16}
\bibitem{journals/ai/JanotaKMC16}
\bibinfo{author}{Mikol{\'{a}}s \surnamestart Janota\surnameend},
  \bibinfo{author}{William \surnamestart Klieber\surnameend},
  \bibinfo{author}{Joao \surnamestart Marques{-}Silva\surnameend} \&
  \bibinfo{author}{Edmund~M. \surnamestart Clarke\surnameend}
  (\bibinfo{year}{2016}): \emph{\bibinfo{title}{Solving {QBF} with
  counterexample guided refinement}}.
\newblock {\sl \bibinfo{journal}{Artif. Intell.}} \bibinfo{volume}{234}, pp.
  \bibinfo{pages}{1--25}, \doi{10.1016/j.artint.2016.01.004}.

\bibitemdeclare{inproceedings}{conf/ijcai/JanotaM15}
\bibitem{conf/ijcai/JanotaM15}
\bibinfo{author}{Mikol{\'{a}}s \surnamestart Janota\surnameend} \&
  \bibinfo{author}{Joao \surnamestart Marques{-}Silva\surnameend}
  (\bibinfo{year}{2015}): \emph{\bibinfo{title}{Solving {QBF} by Clause
  Selection}}.
\newblock In: {\sl \bibinfo{booktitle}{Proceedings of {IJCAI}}},
  \bibinfo{publisher}{{AAAI} Press}, pp. \bibinfo{pages}{325--331}.

\bibitemdeclare{inproceedings}{conf/epia/JanotaM17}
\bibitem{conf/epia/JanotaM17}
\bibinfo{author}{Mikol{\'{a}}s \surnamestart Janota\surnameend} \&
  \bibinfo{author}{Joao \surnamestart Marques{-}Silva\surnameend}
  (\bibinfo{year}{2017}): \emph{\bibinfo{title}{An Achilles' Heel of
  Term-Resolution}}.
\newblock In: {\sl \bibinfo{booktitle}{Proceedings of {EPIA}}}, {\sl
  \bibinfo{series}{LNCS}} \bibinfo{volume}{10423},
  \bibinfo{publisher}{Springer}, pp. \bibinfo{pages}{670--680},
  \doi{10.1007/978-3-319-65340-2_55}.

\bibitemdeclare{inproceedings}{conf/sat/KlieberSGC10}
\bibitem{conf/sat/KlieberSGC10}
\bibinfo{author}{William \surnamestart Klieber\surnameend},
  \bibinfo{author}{Samir \surnamestart Sapra\surnameend},
  \bibinfo{author}{Sicun \surnamestart Gao\surnameend} \&
  \bibinfo{author}{Edmund~M. \surnamestart Clarke\surnameend}
  (\bibinfo{year}{2010}): \emph{\bibinfo{title}{A Non-prenex, Non-clausal {QBF}
  Solver with Game-State Learning}}.
\newblock In: {\sl \bibinfo{booktitle}{Proceedings of {SAT}}}, {\sl
  \bibinfo{series}{LNCS}} \bibinfo{volume}{6175},
  \bibinfo{publisher}{Springer}, pp. \bibinfo{pages}{128--142},
  \doi{10.1007/978-3-642-14186-7_12}.

\bibitemdeclare{inproceedings}{conf/sat/LonsingB08}
\bibitem{conf/sat/LonsingB08}
\bibinfo{author}{Florian \surnamestart Lonsing\surnameend} \&
  \bibinfo{author}{Armin \surnamestart Biere\surnameend}
  (\bibinfo{year}{2008}): \emph{\bibinfo{title}{Nenofex: Expanding {NNF} for
  {QBF} Solving}}.
\newblock In: {\sl \bibinfo{booktitle}{Proceedings of {SAT}}}, {\sl
  \bibinfo{series}{LNCS}} \bibinfo{volume}{4996},
  \bibinfo{publisher}{Springer}, pp. \bibinfo{pages}{196--210},
  \doi{10.1007/978-3-540-79719-7_19}.

\bibitemdeclare{inproceedings}{conf/sat/NiemetzPLSB12}
\bibitem{conf/sat/NiemetzPLSB12}
\bibinfo{author}{Aina \surnamestart Niemetz\surnameend},
  \bibinfo{author}{Mathias \surnamestart Preiner\surnameend},
  \bibinfo{author}{Florian \surnamestart Lonsing\surnameend},
  \bibinfo{author}{Martina \surnamestart Seidl\surnameend} \&
  \bibinfo{author}{Armin \surnamestart Biere\surnameend}
  (\bibinfo{year}{2012}): \emph{\bibinfo{title}{Resolution-Based Certificate
  Extraction for {QBF}}}.
\newblock In: {\sl \bibinfo{booktitle}{Proceedings of {SAT}}}, {\sl
  \bibinfo{series}{LNCS}} \bibinfo{volume}{7317},
  \bibinfo{publisher}{Springer}, pp. \bibinfo{pages}{430--435},
  \doi{10.1007/978-3-642-31612-8_33}.

\bibitemdeclare{inproceedings}{conf/sat/PeitlSS17}
\bibitem{conf/sat/PeitlSS17}
\bibinfo{author}{Tom{\'{a}}s \surnamestart Peitl\surnameend},
  \bibinfo{author}{Friedrich \surnamestart Slivovsky\surnameend} \&
  \bibinfo{author}{Stefan \surnamestart Szeider\surnameend}
  (\bibinfo{year}{2017}): \emph{\bibinfo{title}{Dependency Learning for
  {QBF}}}.
\newblock In: {\sl \bibinfo{booktitle}{Proceedings of {SAT}}}, {\sl
  \bibinfo{series}{LNCS}} \bibinfo{volume}{10491},
  \bibinfo{publisher}{Springer}, pp. \bibinfo{pages}{298--313},
  \doi{10.1007/978-3-319-66263-3_19}.

\bibitemdeclare{inproceedings}{conf/date/PigorschS09}
\bibitem{conf/date/PigorschS09}
\bibinfo{author}{Florian \surnamestart Pigorsch\surnameend} \&
  \bibinfo{author}{Christoph \surnamestart Scholl\surnameend}
  (\bibinfo{year}{2009}): \emph{\bibinfo{title}{Exploiting structure in an
  {AIG} based {QBF} solver}}.
\newblock In: {\sl \bibinfo{booktitle}{Proceedings of {DATE}}},
  \bibinfo{publisher}{IEEE}, pp. \bibinfo{pages}{1596--1601},
  \doi{10.1109/DATE.2009.5090919}.

\bibitemdeclare{article}{journals/jsc/PlaistedG86}
\bibitem{journals/jsc/PlaistedG86}
\bibinfo{author}{David~A. \surnamestart Plaisted\surnameend} \&
  \bibinfo{author}{Steven \surnamestart Greenbaum\surnameend}
  (\bibinfo{year}{1986}): \emph{\bibinfo{title}{A Structure-Preserving Clause
  Form Translation}}.
\newblock {\sl \bibinfo{journal}{J. Symb. Comput.}}
  \bibinfo{volume}{2}(\bibinfo{number}{3}), pp. \bibinfo{pages}{293--304},
  \doi{10.1016/S0747-7171(86)80028-1}.

\bibitemdeclare{inproceedings}{conf/sat/Pulina16}
\bibitem{conf/sat/Pulina16}
\bibinfo{author}{Luca \surnamestart Pulina\surnameend} (\bibinfo{year}{2016}):
  \emph{\bibinfo{title}{The Ninth {QBF} Solvers Evaluation - Preliminary
  Report}}.
\newblock In: {\sl \bibinfo{booktitle}{Proceedings of {QBF@SAT}}}, {\sl
  \bibinfo{series}{{CEUR} Workshop Proceedings}} \bibinfo{volume}{1719},
  \bibinfo{publisher}{CEUR-WS.org}, pp. \bibinfo{pages}{1--13}.

\bibitemdeclare{inproceedings}{conf/fmcad/RabeT15}
\bibitem{conf/fmcad/RabeT15}
\bibinfo{author}{Markus~N. \surnamestart Rabe\surnameend} \&
  \bibinfo{author}{Leander \surnamestart Tentrup\surnameend}
  (\bibinfo{year}{2015}): \emph{\bibinfo{title}{{CAQE:} {A} Certifying {QBF}
  Solver}}.
\newblock In: {\sl \bibinfo{booktitle}{Proceedings of {FMCAD}}},
  \bibinfo{publisher}{IEEE}, pp. \bibinfo{pages}{136--143}.

\bibitemdeclare{inproceedings}{conf/sat/SoosNC09}
\bibitem{conf/sat/SoosNC09}
\bibinfo{author}{Mate \surnamestart Soos\surnameend}, \bibinfo{author}{Karsten
  \surnamestart Nohl\surnameend} \& \bibinfo{author}{Claude \surnamestart
  Castelluccia\surnameend} (\bibinfo{year}{2009}):
  \emph{\bibinfo{title}{Extending {SAT} Solvers to Cryptographic Problems}}.
\newblock In: {\sl \bibinfo{booktitle}{Proceedings of {SAT}}}, {\sl
  \bibinfo{series}{LNCS}} \bibinfo{volume}{5584},
  \bibinfo{publisher}{Springer}, pp. \bibinfo{pages}{244--257},
  \doi{10.1007/978-3-642-02777-2_24}.

\bibitemdeclare{inproceedings}{conf/sat/Tentrup16}
\bibitem{conf/sat/Tentrup16}
\bibinfo{author}{Leander \surnamestart Tentrup\surnameend}
  (\bibinfo{year}{2016}): \emph{\bibinfo{title}{Non-prenex {QBF} Solving Using
  Abstraction}}.
\newblock In: {\sl \bibinfo{booktitle}{Proceedings of {SAT}}}, {\sl
  \bibinfo{series}{LNCS}} \bibinfo{volume}{9710},
  \bibinfo{publisher}{Springer}, pp. \bibinfo{pages}{393--401},
  \doi{10.1007/978-3-319-40970-2_24}.

\bibitemdeclare{inproceedings}{conf/cav/Tentrup17}
\bibitem{conf/cav/Tentrup17}
\bibinfo{author}{Leander \surnamestart Tentrup\surnameend}
  (\bibinfo{year}{2017}): \emph{\bibinfo{title}{On Expansion and Resolution in
  {CEGAR} Based {QBF} Solving}}.
\newblock In: {\sl \bibinfo{booktitle}{Proceedings of {CAV}}}, {\sl
  \bibinfo{series}{LNCS}} \bibinfo{volume}{10427},
  \bibinfo{publisher}{Springer}, pp. \bibinfo{pages}{475--494},
  \doi{10.1007/978-3-319-63390-9\_25}.

\bibitemdeclare{inproceedings}{conf/sat/TuHJ15}
\bibitem{conf/sat/TuHJ15}
\bibinfo{author}{Kuan{-}Hua \surnamestart Tu\surnameend},
  \bibinfo{author}{Tzu{-}Chien \surnamestart Hsu\surnameend} \&
  \bibinfo{author}{Jie{-}Hong~R. \surnamestart Jiang\surnameend}
  (\bibinfo{year}{2015}): \emph{\bibinfo{title}{{QELL:} {QBF} Reasoning with
  Extended Clause Learning and Levelized {SAT} Solving}}.
\newblock In: {\sl \bibinfo{booktitle}{Proceedings of {SAT}}}, {\sl
  \bibinfo{series}{LNCS}} \bibinfo{volume}{9340},
  \bibinfo{publisher}{Springer}, pp. \bibinfo{pages}{343--359},
  \doi{10.1007/978-3-319-24318-4_25}.

\end{thebibliography}

\end{document}